\documentclass[preprint,authoryear,12pt]{elsarticle}

\usepackage{latexsym}
\usepackage{amssymb}
\usepackage{mathrsfs}
\usepackage{amsmath}
\usepackage{algorithm}
\usepackage{algorithmic}

\usepackage{graphicx}
\usepackage{colortbl}
\usepackage{color}

\usepackage{geometry}
\usepackage{enumerate}
\usepackage{bm} 
\usepackage{bbm} 
\usepackage{booktabs}
\usepackage[hang]{subfigure} 

\usepackage[
setpagesize=false, 
pdfborder={0 0 0}, 
pdfpagemode=UseOutlines, 
]{hyperref} 

\usepackage{rotating}
\usepackage{lscape}
\usepackage[latin1]{inputenc}

\numberwithin{algorithm}{section}

\newcommand{\bbr}{\mathbb{R}}
\newcommand{\bbc}{\mathbb{C}}
\newcommand{\bbh}{\mathbb{H}}

\newcommand{\bx}{\boldsymbol{x}}
\newcommand{\btheta}{\boldsymbol{\theta}}
\newcommand{\bX}{\boldsymbol{X}}
\newcommand{\bu}{\boldsymbol{u}}
\newcommand{\bU}{\boldsymbol{U}}

\newcommand{\calT}{{\cal T}}
\newcommand{\calt}{{t}}

\newcommand{\calc}{{\cal C}}

\usepackage[
thmmarks, 
amsmath, 
hyperref 
]{ntheorem} 

\theoremstyle{break} 
\theoremheaderfont{\bfseries}
\theorembodyfont{\itshape}
\theoremseparator{}
\newtheorem{definition}{Definition} 

\newtheorem{theorem}[definition]{Theorem}

\newtheorem{example}{Example}
\theoremstyle{nonumberbreak} 
\theoremsymbol{$\Box$}
\newtheorem{proof}{Proof}

\usepackage{caption}
\usepackage{capt-of}
\makeatletter
\renewcommand*\env@matrix[1][*\c@MaxMatrixCols c]{%
  \hskip -\arraycolsep
  \let\@ifnextchar\new@ifnextchar
  \array{#1}}
\makeatother

\journal{}

\begin{document}

\begin{frontmatter}

\title{Efficient goodness-of-fit tests in multi-dimensional vine copula models}

	\author{Ulf Schepsmeier }

\address{Center for Mathematical
Sciences, Technische Universit\"at M\"unchen, Boltzmannstraße 3,
85748 Garching b. M\"unchen, Germany. Email: \texttt{schepsme@ma.tum.de}, Phone: +49.89.289.17422}

\begin{abstract}
We introduce a new goodness-of-fit test for regular vine (R-vine) copula models, a flexible class of multivariate copulas based on a pair-copula construction (PCC). The test arises from the information matrix ratio.
The corresponding test statistic is derived and its asymptotic normality is proven. The test's power is investigated and compared to 14 other goodness-of-fit tests, adapted from the bivariate copula case, in a high dimensional setting. The extensive simulation study shows the excellent performance with respect to size and power as well as the superiority of the information matrix ratio based test against most other goodness-of-fit tests. The best performing tests are applied to a portfolio of stock indices and their related volatility indices validating different R-vine specifications.
\end{abstract}


\end{frontmatter}

\section{Introduction}
\label{sec:introduction}

Analyzing complex correlated data has received considerable attention
in the current statistical literature. 
Among many approaches to modeling correlation structures, copula based models offer a powerful and flexible toolbox to characterize dependence profiles among variables, which have been studied extensively. 
However, it is unfortunate that there is little progress known in the theory and method concerning a goodness-of-fit (GOF) test, an important aspect of statistical model diagnostics. In fact, most of the published work has been only focused on bivariate copula models \citep[see for example][]{Genest2009}.

Copulas join marginal distributions $F_1,\ldots,F_d$ of a (continuous) random vector $\bX=(X_1,\ldots,X_d)$ with their dependency structure by a joint cumulative distribution function (cdf) $H(x_1,\ldots,x_d)=C(F_1(x_1),\ldots,F_d(x_d))$. Here $C$ is the unique cdf with uniform margins on the unit hypercube \citep{Sklar}. Classical copula classes such as the elliptical or Archimedean copulas are very limited with respect to flexibility in higher dimensions. But they are very powerful and well understood in the bivariate case. Thus \citet{Joe2} and later \citet{Bedford_Cooke2001,Bedford_Cooke} independently constructed multivariate densities using $d(d-1)/2$ bivariate copulas. They permit flexibility and feasibility of constructing and computing a relatively large dimensional copula model.
In \citet{Aas_Czado} this process is termed a pair-copula construction (PCC) and the statistical inference is developed for it. Since then the theory of vine copulas arising from the PCC were studied in literature  \citep[see for example][]{Czado,StoeberSchepsmeier2012,CzadoSchepsmeierMin2011,DissmannBrechmannCzadoKurowicka2011}.

Along with the break through of vine copula constructions model diagnosis becomes ever so imperative in the application of multi-dimensional vine copulas.
Developing efficient GOF tests is now a timely task as already noted in \citet{Fermanian2012}, and an
important addition to the current literature of vine copulas. In
addition, comprehensive comparisons for many of the classical
GOF tests are lacking in terms of their relative merits when they are
applied to multi-dimensional copulas.
So far model verification methods for vine copulas are usually based on the likelihood, or on the Akaike Information Criterion (AIC) or Bayesian Information Criterion (BIC) as classical comparison measures, which take the model complexity into account. 

In our goodness-of-fit (GOF) tests we would like to test
\begin{equation}
	H_0: C\in \calc_0=\{C_{\btheta}:\btheta\in\Theta\}\quad \text{against}\quad H_1: C\notin\calc_0=\{C_{\btheta}:\btheta\in\Theta\},
	\label{eq:test2}
\end{equation}
where $C$ denotes the (vine) copula distribution function and $\calc_0$ is a class of parametric (vine) copulas
with $\Theta\subseteq\bbr^p$ being the parameter space of dimension $p$.

For the elliptical and parametric Archimedean copulas many GOF tests were studied in the literature \citep{Genest_Remillard_2,Genest2009,Berg,HuangProkhorov2011}.
However, a GOF test for vine copula models verifying the chosen pair-copula families has, to our knowledge, only be treated in \citet{Schepsmeier2013}. 
Although, already \citet{Aas_Czado} suggested a GOF test for vine copulas based on the multivariate probability integral transformation (PIT) of \citet{Rosenblatt1952} given in the appendix, but never investigated its small sample performance. We will show that this test and many other copula GOF tests have little to no power in the high dimensional setting of a vine and thus are not appropriate to be utilized there.

The main contribution of this paper is a new GOF test to perform model verification of vine copula models using hypothesis tests. As in \citet{Schepsmeier2013} it is based on the Bartlett identity ($-\bbh(\btheta)=\bbc(\btheta)$) as generally suggested by \citet{White1982}. Here $\bbh(\btheta)$ is the expected Hessian or variability matrix, and $\bbc(\btheta)$ is the expected outer product of the gradient or sensitivity matrix.
In contrast to the White test, which relies on the difference between $-\bbh(\btheta)$ and $\bbc(\btheta)$, our new test is based on the information matrix ratio (IMR), $\Psi(\btheta) = -\bbc(\btheta)^{-1}\bbh(\btheta)$ \citep{ZhouSongThompson2012}. 

First, the IMR based test statistic for vine models will be derived and its asymptotic normality under the Bartlett identity will be proven. Secondly, the small sample performance for size and power will be investigated and compared to 14 other GOF tests for vines in a high dimensional setting ($d=5$ and $d=8$). In particular, we will compare to GOF tests based on the
\begin{itemize}
	\item difference of Bartlett identity or
	\item empirical copula process $\hat{C}_n(\bu)-C_{\hat{\btheta}_n}(\bu)$, with $\bu=(u_1,\ldots,u_d)\in[0,1]^d$,
	\begin{equation}
	\hat{C}_n(\bu) = \frac{1}{n+1}\sum_{t=1}^n \boldsymbol{1}_{\{U_{t1}\leq u_1,\ldots,U_{td}\leq u_d \}},
	\label{eq:empCopula}
\end{equation}
	and $C_{\hat{\btheta}_n}(\bu)$ being the copula with estimated parameter(s) $\hat{\btheta}_n$, and/or
	\item multivariate PIT.
\end{itemize}
For the tests based on the multivariate PIT aggregation to univariate test data is facilitated using different aggregation functions. For the univariate test data then standard univariate GOF test statistics such as Anderson-Darling (AD), Cramér-von Mises (CvM) and Kolmogorov-Smirnov (KS) are used.
In contrast, the empirical copula process (ECP) based test use the multivariate Cramér-von Mises (mCvM) and multivariate Kolmogorov-Smirnov (mKS) test statistics. The different GOF tests are given in the appendix for the convenience of the reader.

The power study will expose that the information based GOF tests such as the information matrix difference approach of \citet{Schepsmeier2013} and in particular our new IMR based test outperform the other GOF tests in terms of size and power. The PIT based GOF tests reveal little to no power against the considered alternatives. But applying the PIT transformed data to the empirical copula process, as first suggested by \citet{Genest2009}, is more promising. Here $C_{\hat{\btheta}_n}(\bu)$ is replaced by the independence copula $C_{\bot}$ in the ECP.

The remainder of this paper is structured as follows: Section \ref{sec:rvine} gives an introduction on vine copula models. The new proposed IR test is introduced and its test statistics derived in Section \ref{sec:IR}. Additionally the asymptotic normality of the test statistic is proven. Further GOF tests extended from known copula GOF tests are given in Section \ref{sec:gof} for the extensive power comparison study in Section \ref{sec:powerstudy} investigating their size and power. An application of an 8-dimensional portfolio of stock indices and their related volatility indices is performed in Section \ref{sec:application} comparing different vine specifications and proposed GOF tests. The final Section \ref{sec:discussion} summarizes and shows areas of further research.

\section{Regular vine copula model}
\label{sec:rvine}

Pair-copula constructions (PCC) are a very flexible way to model multivariate distributions with copulas. The model is based on the decomposition of the $d$-dimensional density into $d(d-1)/2$ (conditional) bivariate copula densities.
\citet{Bedford_Cooke2001,Bedford_Cooke} introduced linked trees $T_i=(V_i, E_i)$, where $V_i$ denotes the set of nodes while $E_i$ represents the set of edges, which helps to organize the vine construction. The following conditions have to be fulfilled to call a sequence of trees $\mathcal{V}=(T_1,\ldots,T_{d-1})$ a vine:
\begin{enumerate}
	\item $T_1$ is a tree with nodes $V_1=\{1,\ldots,d\}$ and edges $E_1$.
	\item For $i\geq 2$, $T_i$ is a tree with nodes $V_i = E_{i-1}$ and edges $E_i$.
	\item If two nodes in tree $T_{i+1}$ are joined by an edge, the corresponding edges in $T_i$ must share a common node {\it (proximity condition)}.
\end{enumerate}
An example of a vine tree sequence $\mathcal{V}$ is given in Figure \ref{fig:5dimRvine}. Here $V_1=\{1,\ldots,5\}$ and $E_1=\{\{1,2\},\{1,3\},\{1,4\},\{4,5\}\}$ forming the unconditional pair-copulas. The conditional pair-copulas are the edges of tree 2-4. The complete construction of the joint density is given in Example \ref{example}.

Following the notation of \citet{Czado} we define a set of bivariate copula densities  $\mathcal{B} = \left\{c_{j(e),k(e);D(e)} \vert e \in E_i, 1 \leq i \leq d-1 \right\}$ corresponding to edges $j(e),k(e)|D(e)$ in $E_i$, for $1\leq i\leq d-1$. Here $\boldsymbol{u}_{D(e)}$ denotes the subvector of $\bu=(u_1,\ldots,u_d)\in[0,1]^d$ determined by the set of indices in $D(e)$. The set $D(e)$ is called the {\it conditioning set} while the indices $j(e)$ and $k(e)$ form the {\it conditioned set}. Then a $d$-dimensional regular vine copula density can be constructed as
\begin{equation}
	c_{1,\ldots,d}(\bu)
	=\prod_{i=1}^{d-1}\prod_{e\in E_i}c_{j(e),k(e);D(e)}(C_{j(e)|D(e)}(u_{j(e)}|\bu_{D(e)}),C_{k(e)|D(e)}(u_{k(e)}|\bu_{D(e)})).
	\label{eq:density}
\end{equation}
For the copula arguments, the conditional cdfs
$C_{j(e)|D(e)}(u_{j(e)}|\bu_{D(e)})$ and \\
$C_{k(e)|D(e)}(u_{k(e)}|\bu_{D(e)})$, \citet{Joe2} developed a formula derived from the first derivative of the corresponding cdf with respect to the second copula argument, i.e.
{\small
\begin{gather}
\begin{split}
	C_{j(e)|D(e)}(u_{j(e)}|\bu_{D(e)}) = \frac{\partial C_{j(e),j^{\prime}(e);D(e)\setminus j^{\prime}(e)}(C(u_{j(e)}|\bu_{D(e) \setminus j^{\prime}(e)}),C(u_{j^{\prime}(e)}|\bu_{D(e)\setminus j^{\prime}(e)}))}{\partial C(u_{j^{\prime}(e)}|\bu_{D(e)\setminus j^{\prime}(e)})}.
\end{split}
\label{eq:hfunction}
\end{gather}}

Here $j^{\prime}(e)\in D(e)$ is an index chosen from the conditioning set, such that $C_{j(e),j^{\prime}(e)|D(e)\setminus j^{\prime}(e)}$ is in $\mathcal{B}$.
In the literature Equation (\ref{eq:hfunction}) is often called a {\it h-function}. It is a recursive function which simplifies the calculation of the density or log-likelihood considerably. See for example \citet{DissmannBrechmannCzadoKurowicka2011} for a algorithmic presentation of the log-likelihood of an R-vine. Denoting the pair-copula parameters of $\mathcal{B}$ with $\btheta=\btheta(\mathcal{B}(\mathcal{V}))$ a vine copula model with density given in (\ref{eq:density}) is abbreviated as $RV=(\mathcal{V},\mathcal{B}(\mathcal{V}),\btheta(\mathcal{B}(\mathcal{V})))$.
We assume that the copula $C_{j(e),j'(e);D(e)\setminus j'(e)}$ does not depend on the values $\bu_{D(e)\setminus j^{\prime}(e)}$, i.e.~on the conditioning set without the chosen variable $u_{j^{\prime}(e)}$. This is called the {\it simplifying assumption}.


\begin{example}[5-dim pair-copula construction]
\label{example}
The corresponding copula density to the vine tree sequence given in Figure \ref{fig:5dimRvine} can be expressed as
\begin{gather}
\begin{split}
	&c_{12345}(u_1,\ldots,u_5) =  c_{1,2}(u_1,u_2)\cdot c_{1,3}(u_1,u_3)\cdot c_{1,4}(u_1,u_4)\cdot c_{4,5}(u_4,u_5) \\
	&\cdot c_{2,4;1}(C_{2|1}(u_2|u_1),C_{4|1}(u_4|u_1))\cdot c_{3,4;1}(C_{3|1}(u_3|u_1),C_{4|1}(u_4|u_1)) \\
	&\cdot c_{1,5;4}(C_{1|4}(u_1|u_4),C_{5|4}(u_5|u_4)) \\
	&\cdot c_{2,3;1,4}(C_{2|1,4}(C_{2|1}(u_2|u_1),C_{4|1}(u_4|u_1)),C_{3|1,4}(C_{3|1}(u_3|u_1),C_{4|1}(u_4|u_1))) \\
	&\cdot c_{3,5;1,4}(C_{3|1,4}(C_{3|1}(u_3|u_1),C_{4|1}(u_4|u_1)),C_{5|1,4}(C_{1|4}(u_1|u_4),C_{5|4}(u_5|u_4))) \\
	&\cdot c_{2,5;1,3,4}(C_{2|1,3,4}(C_{2|1,4}(C_{2|1}(u_2|u_1),C_{4|1}(u_4|u_1)),C_{3|1,4}(C_{3|1}(u_3|u_1),C_{4|1}(u_4|u_1))),\\
	&\qquad\qquad C_{5|1,3,4}(C_{3|1,4}(C_{3|1}(u_3|u_1),C_{4|1}(u_4|u_1)),C_{5|1,4}(C_{1|4}(u_1|u_4),C_{5|4}(u_5|u_4)))).
 \label{eq:5dimRvine}
\end{split}
\end{gather}
\end{example}

\begin{figure}[htbp]
	\centering
		\includegraphics[width=0.9\textwidth]{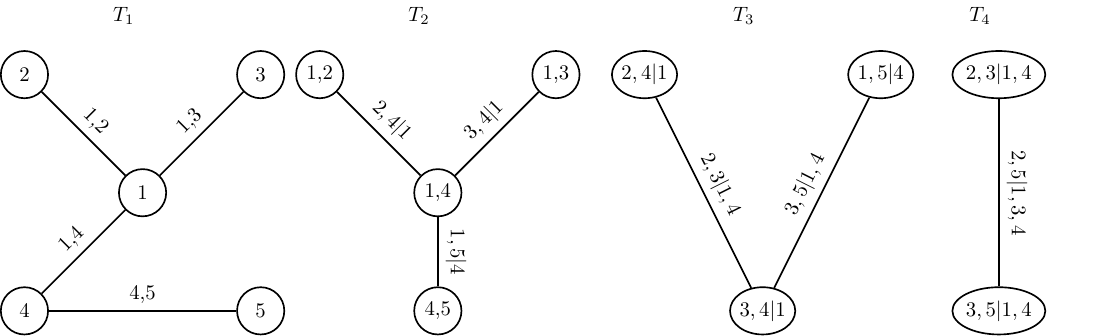}
	\caption{Tree structure of the 5 dimensional R-vine copula given in Example \ref{example}.}
	\label{fig:5dimRvine}
\end{figure}

There are two special cases of an R-vine tree structure $\mathcal{V}$. A line like structure of the trees is called D-vine in which each node has a maximum degree of 2, while a star structure is a canonical vine (C-vine) with a root node of degree $d-1$. All other nodes have degree 1. Statistical inference methods of D-vines are discussed in \citet{Aas_Czado}. A model selection algorithm as well as the maximum likelihood parameter estimation for C-vines is developed in \citet{CzadoSchepsmeierMin2011}.

\section{Information matrix ratio test}
\label{sec:IR}

A new approach for a GOF test for vine copulas is the information ratio (IR) test. It is inspired by the paper of \citet{ZhouSongThompson2012}, who propose an IR test for general model misspecification of the variance or covariance structures. Their test is related to the ``in-and-out-sample'' (IOS) test of \citet{PresnellBoos2004}, which is a likelihood ratio test. Additionally \citet{PresnellBoos2004} showed that the IOS test statistic can be expressed as a ratio of the expected Hessian and the expected outer product of the gradient. 
In particular, let $\bU=(U_1,\ldots,U_d)^T\in[0,1]^d$ be a random vector with copula distribution function $C_{\btheta}(u_1,\ldots,u_d)$. Further let
\begin{equation}
	\bbh(\btheta) := E\left[ \partial_{\btheta}^2 l(\btheta|\bU)  \right]\quad\text{and}\quad
	\bbc(\btheta) := E\left[ \partial_{\btheta} l(\btheta|\bU) \big( \partial_{\btheta} l(\btheta|\bU) \big)^T \right]
\label{eq:HC}
\end{equation}
the expected Hessian matrix of the random (vine) copula log-likelihood function $l(\btheta|\bU):=ln(c_{\btheta}(U_1,\ldots,U_d))$ and the expected outer product of the corresponding score function, respectively. Here $\partial_{\btheta}$ denotes the derivative with respect to the copula parameter $\btheta\in\bbr^{p}$. 
Now the information matrix ratio (IMR) is defined as
\begin{equation}
\Psi(\btheta) := -\bbh(\btheta)^{-1}\bbc(\btheta).
\label{eq:IMR}
\end{equation}
Our test problem is the reformulated general test problem
of \citet{White1982}:
\[
	H_0: \Psi(\btheta)=I_p\quad \text{against}\quad H_1: \Psi(\btheta)\neq I_p,
\]
where $I_p$ is the $p$-dimensional identity matrix.
To calculate the corresponding test statistic we follow \citet{Schepsmeier2013} and define the random matrices
\begin{equation}
	\bbh(\btheta|\bU):=\frac{\partial^2}{\partial^2\btheta}l(\btheta|\bU)\qquad\text{and}\qquad \bbc(\btheta|\bU):=\frac{\partial}{\partial\btheta}l(\btheta|\bU)\left(\frac{\partial}{\partial\btheta}l(\btheta|\bU)\right)^T
	\label{eq:HC2}
\end{equation}
using the log-likelihood function $l(\btheta|\bU)$ of the $\bU\sim RV(\mathcal{V},\mathcal{B}(\mathcal{V}),\btheta(\mathcal{B}(\mathcal{V})))$ model with specified vine tree sequence $\mathcal{V}$ and pair-copulas $\mathcal{B}(\mathcal{V})$ but unknown parameter $\btheta=\btheta(\mathcal{B}(\mathcal{V}))$.
Given an i.i.d.~sample $\bu_t\in[0,1]^d$ from $RV(\mathcal{V},\mathcal{B}(\mathcal{V}),\btheta(\mathcal{B}(\mathcal{V})))$ for $t=1,\ldots,n$ and the corresponding maximum likelihood estimate $\hat{\btheta}_n$ based on $\bu=(\bu_1^T,\ldots,\bu_n^T)^T$ the sample counter parts are
\[
	\hat{\bbh}_t(\hat{\btheta}_n):=\bbh(\hat{\btheta}_n|\bu_t)\in \bbr^{p\times p}\qquad\text{and}\qquad \hat{\bbc}_t(\hat{\btheta}_n):=\bbc(\hat{\btheta}_n|\bu_t)\in \bbr^{p\times p},
\]
respectively.
The sample equivalents to $\bbh(\btheta)$ and $\bbc(\btheta)$ are then
\begin{equation}
	\bar{\bbh}(\hat{\btheta}_n) := \frac{1}{n}\sum_{t=1}^n \hat{\bbh}_t(\hat{\btheta}_n)\qquad\text{and}\qquad
	\bar{\bbc}(\hat{\btheta}_n) := \frac{1}{n}\sum_{t=1}^n \hat{\bbc}_t(\hat{\btheta}_n).
	\label{eq:HCemp}
\end{equation}
Thus, we get as empirical version of (\ref{eq:IMR}): $\bar{\Psi}(\hat{\btheta}_n) := -\bar{\bbh}(\hat{\btheta}_n)^{-1}\bar{\bbc}(\hat{\btheta}_n).$\\
As in \citet{ZhouSongThompson2012} we define the {\it information ratio} \textbf{(IR)} statistic as
\begin{equation}
	IR_n := tr(\bar{\Psi}(\hat{\btheta}_n))/p,
\label{eq:IR}
\end{equation}
where $tr(A)$ denotes the trace of matrix $A$.
To derive the asymptotic normality of the test statistic $IR_n$ some conditions have to be set. The first two conditions $\mathcal{C}_1$ and $\mathcal{C}_2$ guarantee the existence of the gradient and the Hessian matrix.
\begin{itemize}
	\item[$\mathcal{C}_1:$] The density function (\ref{eq:density}) is twice continuous differentiable with respect to $\btheta$.
	\item[$\mathcal{C}_2:$] -$\bar{\bbh}(\hat{\btheta}_n)$ and $\bar{\bbc}(\hat{\btheta}_n)$ are positive definite.
\end{itemize}
Condition $\mathcal{C}_3-\mathcal{C}_5$ are more technical and are the same as in \citet{PresnellBoos2004}.
\begin{itemize}
	\item[$\mathcal{C}_3:$] There exist $\btheta_0$ such that $\hat{\btheta}_n\stackrel{P}{\rightarrow}\btheta_0$ as $n\rightarrow\infty$. 
	\item[$\mathcal{C}_4:$] The estimator $\hat{\btheta}_n\in\bbr^p$ has an approximating influence curve function $h(\btheta|\bu)$ such that
	\[
		\hat{\btheta}_n-\btheta = \frac{1}{n}\sum_{i=1}^n h(\btheta_0|U_i)+R_{n1},
	\]
	where $\sqrt{n}R_{n1}\stackrel{P}{\rightarrow}0$ as $n\rightarrow\infty$, $E[h(\btheta_0|U_1]=0$, and $cov(h(\btheta_0|U_1)$ is finite.
	\item[$\mathcal{C}_5:$] The real-valued function $q(\btheta|\bu)$ possesses second order partial derivatives with respect to $\btheta$, and	
\begin{enumerate}
	\item[(a)] $Var(q(\btheta_0|U_1))$ and $E\left[\frac{\partial}{\partial \btheta}q(\btheta_0|U_1)\right]$ are finite.
	\item[(b)] There exists a function $M(\bu)$ such that for all $\btheta$ in a neighborhood of $\btheta_0$ and all $j,k\in\{1,\ldots,p\}, \left|\frac{\partial^2}{\partial^2\btheta}q(\btheta|\bu)_{jk}\right|\leq M(\bu)$, where $E[M(U_1)]<\infty$. 
\end{enumerate}
\end{itemize}
In the following $vech(A)\in\bbr^{p(p+1)/2}$ represents the vectorization of the symmetric matrix $A\in\bbr^{p\times p}$. Let $\boldsymbol{W}:=(W_1,\ldots,W_{p(p+1)})^T = (vech(\bar{\bbc}(\hat{\btheta}_n)),vech(\bar{\bbh}(\hat{\btheta}_n)))^T \in\bbr^{p(p+1)}$, then \citet{PresnellBoos2004} showed that
\[
	\sqrt{n}\frac{\boldsymbol{W}-\boldsymbol{\mu}_W}{\Sigma_W^{1/2}} \stackrel{d}{\rightarrow} N_{p(p+1)}(\boldsymbol{0}_{p(p+1)},I_{p(p+1)}),
\]
where $\boldsymbol{\mu}_W$ is the mean vector and $\Sigma_W$ is the asymptotic covariance matrix of $\boldsymbol{W}$. Here $\boldsymbol{0}_{p(p+1)}:=(0,\ldots,0)^T$ is the $p(p+1)$-dimensional zero vector and $I_{p(p+1)}$ is the $p(p+1)$-dimensional identity matrix.
Furthermore, let $D(\hat{\btheta}_n)$ define the partial derivatives of $IR_n$ taken with respect to the components of $\boldsymbol{W}$, i.e.
\[
	D(\hat{\btheta}_n):=\left( \frac{\partial IR_n}{\partial W_i} \right)_{i=1,\ldots,p(p+1)} \in\bbr^{p(p+1)}.
\]

\begin{theorem}
Let $\bU\sim RV(\mathcal{V},\mathcal{B}(\mathcal{V}),\btheta(\mathcal{B}(\mathcal{V})))$ satisfy the conditions $\mathcal{C}_1-\mathcal{C}_3$. 
Further, let $\mathcal{C}_4$ hold for the maximum likelihood estimator $\hat{\btheta}_n$ with $h(\btheta_0|\bu):=\bbc(\btheta_0)^{-1}\frac{\partial}{\partial\btheta}l(\btheta_0|\bu)$. 
Additionally, the condition $\mathcal{C}_5$ has to be satisfied for both $q(\btheta|\bu):=-\frac{\partial^2}{\partial^2\btheta}l(\btheta|\bu)_{jk}$ and $q(\btheta|\bu):=\left(\frac{\partial}{\partial\btheta}l(\btheta|\bu)\left(\frac{\partial}{\partial\btheta}l(\btheta|\bu)\right)^T\right)_{jk}$ for each $j,k\in\{1,\ldots,p\}$.
Then the IR test statistic
\[
	Z_n := \frac{IR_n-1}{\sigma_{IR}} \stackrel{D}{\rightarrow} N(0,1) \text{ as } n\rightarrow \infty,
\]
where $\sigma_{IR}$ is the standard error of the IR test statistic, defined as
\[
	\sigma_{IR}^2 := \frac{1}{n}D^T\Sigma_W D.
\]
Here $\Sigma_W/n$ is the asymptotic covariance matrix arising from the joint asymptotic normality of $vech(\bar{\bbc}(\hat{\btheta}_n))$ and $vech(\bar{\bbh}(\hat{\btheta}_n))$ defined above.
By $D$ we denote the $p(p+1)$-dimensional vector of partial derivatives of $IR_n$ taken with respect to the components of $\boldsymbol{W}$ and evaluated at their limits in probability, i.e.~$
	D:=D(\hat{\btheta}_n)|_{\hat{\btheta}_n\stackrel{P}{\rightarrow} \btheta_0}$.
\label{prop:asy}
\end{theorem}

\begin{proof}
The proof follows directly from the proof of Theorem 3 in \citet{PresnellBoos2004}, since we have a fully specified likelihood and the conditions of Theorem 3 are assumed to be satisfied for vine copulas considered in Theorem \ref{prop:asy}.
\end{proof}

Since the theoretical asymptotic variance $\sigma_{IR}^2$ is quite difficult to compute, an empirical version is used in practice.
To evaluate the standard error $\sigma_{IR}$ numerically, \citet{ZhouSongThompson2012} suggest a perturbation resampling approach.
Furthermore, \citet{PresnellBoos2004} state that the convergence to normality is slow and thus they suggest obtaining p-values using a parametric bootstrap under the null hypothesis.

The condition $\mathcal{C}_4$ for $q(\btheta|\bu):=-\frac{\partial^2}{\partial^2\btheta}l(\btheta|\bu)_{jk}$ implies, that the copula density function (\ref{eq:density}) is four times differentiable with respect to $\btheta$. Furthermore, the first and second moment of the second derivative has to be finite. The vine copula density is four times differentiable if all selected pair-copulas are four times differentiable. These assumptions are satisfied for the elliptical Gauss and Student's t-copula as well as for the parametric Archimedean copulas in all dimensions.


Let $\alpha\in(0,1)$ and $Z_n$ as in Theorem \ref{prop:asy}. Then the test
	\begin{equation*}
		\text{Reject } H_0: \Psi(\btheta)=I_p\quad \text{against}\quad \Psi(\btheta)\neq I_p\qquad
		\Leftrightarrow\quad Z_n > \Phi^{-1}(1-\alpha)
	\end{equation*}
	is an asymptotic $\alpha$-level test. Here $\Phi^{-1}$ denotes the quantile of a $N(0,1)$-distribution.

\section{Further goodness-of-fit tests for vine copulas}
\label{sec:gof}

In the recent years many GOF test were suggested for copulas. The most promising ones were investigated in \citet{Genest2009} and \citet{Berg}. However only the size and power of the elliptical and one-parametric Archimedean copulas for $d\in\{2,4,8\}$ were analyzed. The multivariate case is therefore poorly addressed. For vine copulas little is done. A first test for vine copulas was suggested but not investigated in \citet{Aas_Czado}. Their GOF is based on the multivariate PIT and an aggregation introduced by \citet{Breymann2003}. After aggregation standard univariate GOF tests such as the Anderson-Darling (AD), the Cramér-von Mises (CvM) or the Kolmogorov-Smirnov (KS) tests are applied. They are described in more detail in \ref{appendix:tests}. We will denote the resulting tests as \textbf{Breymann}. 

Similar approaches based on the multivariate PIT are proposed by \citet{BergBakken2007}. 
Beside new aggregation functions forming univariate test data, they perform the aggregation step on the ordered PIT output data $\boldsymbol{y}_{(1)}^T,\ldots,\boldsymbol{y}_{(d)}^T$ instead of $\boldsymbol{y}_1^T,\ldots,\boldsymbol{y}_d^T$. Again standard univariate GOF tests are applied. These approaches will be called \textbf{Berg} and \textbf{Berg2}, respectively.

\citet{BergAas2009} applied a test for $H_0: C\in \calc_0$ against $H_1: C\notin\calc_0$ based on the empirical copula process (ECP) to a 4-dimensional vine copula. As the Breymann test, their GOF test is not described in detail or investigated with respect to its power. We will denote this test as \textbf{ECP}.
An extension of the ECP-test is the combination of the multivariate PIT approach with the ECP. The general idea is that the transformed data of a multivariate PIT should be ``close'' to the independence copula $C_{\bot}$ \citet{Genest2009}. Thus a distance of CvM or KS type between them is considered. This approach is called \textbf{ECP2}.

\citet{Schepsmeier2013} was the first who analyzed the power of a GOF test for vine copulas in detail. His approach is, as our new IR GOF test, based on the information matrix equality and specification test introduced by \citet{White1982}. His power studies show, that the convergence to the asymptotic distribution function of the test statistic is very slow. Further, given copula data with sample size smaller than 10000 the test does not reach its nominal level based on asymptotic p-values. But using bootstrapped p-values the test shows very good power behavior. We denote this approach as \textbf{White}.

In the forthcoming sections we will introduce the vine copula test of \citet{Schepsmeier2013}, the multivariate PIT based GOF such as the ones of \citet{Breymann2003} and \citet{BergBakken2007}, and the two ECP based GOF tests.
A first overview of the considered GOF tests is given in Figure \ref{fig:diagramm}. 

\begin{sidewaysfigure}[htbp]
	\centering
		\includegraphics[width=1.00\textwidth]{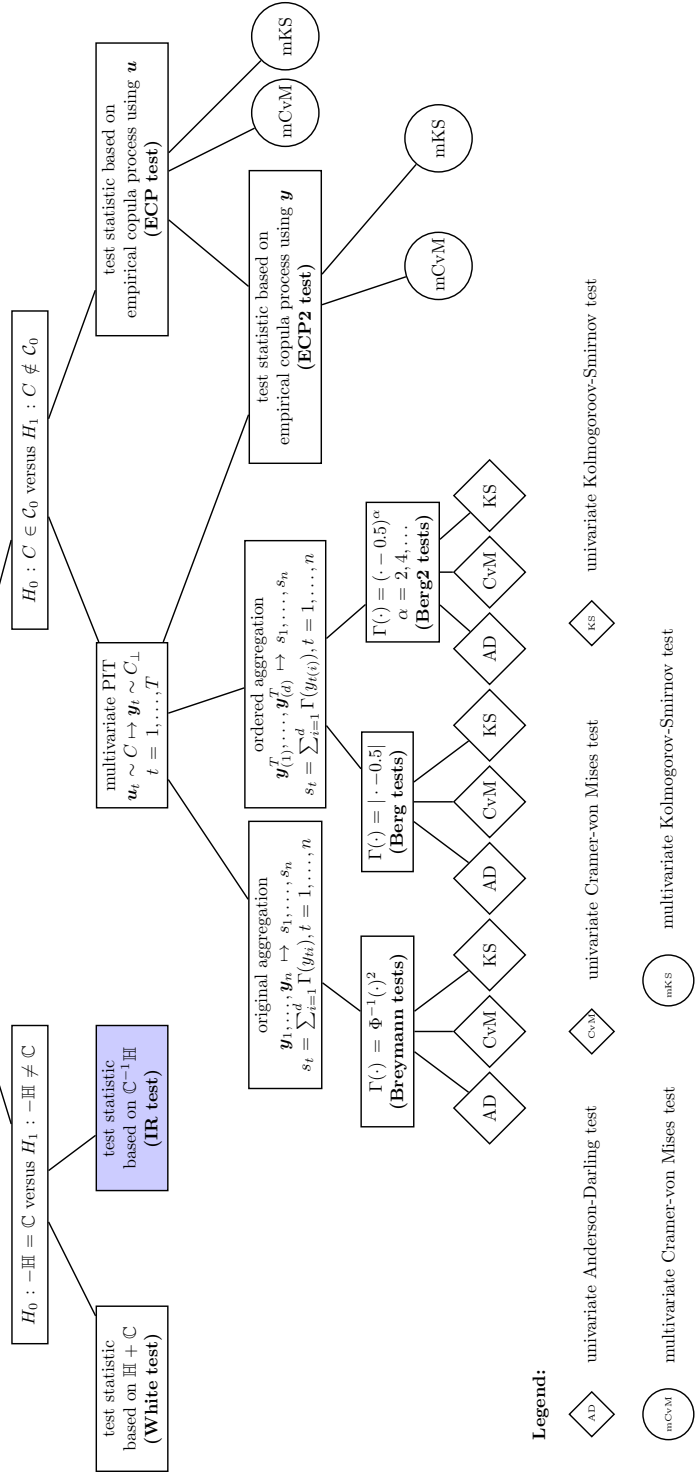}
	\caption{Structured overview of the suggested goodness-of-fit test hypotheses and their test statistics.}
	\label{fig:diagramm}
\end{sidewaysfigure}

\noindent

\subsection{White's information matrix test}
\label{subsec:White}

The GOF test of \citet{Schepsmeier2013} uses White's information matrix equality and specification test. 
It is a rank-based 
test which is asymptotically pivotal, i.e. the asymptotic distribution is independent of model parameters.

Let $\bU$ be a random vector with vine copula log-likelihood $l(\btheta|\bU)$. 
Further let $\bbh(\btheta)$ and $\bbc(\btheta)$ be defined as in Equation (\ref{eq:HC})
 the expected Hessian matrix and the expected outer product of the score function, respectively.
Considering the Bartlett identity
we can formulate the vine copula misspecification test problem as
\begin{equation}
	H_0: \bbh(\btheta_0) + \bbc(\btheta_0) = 0 \text{ against } H_1: \bbh(\btheta_0) + \bbc(\btheta_0) \neq 0.
	\label{eq:test}
\end{equation}
Here, $\btheta_0$ denotes the true value of the vine copula parameter vector.
Following the notation of \citet{Schepsmeier2013} we denote by 
$
\boldsymbol{d}(\btheta|\bU) := vech(\bbh(\btheta|\bU) + \bbc(\btheta|\bU)) \in \bbr^{\frac{p(p+1)}{2}},
$
the vectorized sum of $\bbh(\btheta|\bU)$ and $\bbc(\btheta|\bU)$ defined in (\ref{eq:HC2}). Its empirical version is denoted by $\bar{\boldsymbol{d}}(\hat{\btheta}_n):=vech(\bar{\bbh}(\hat{\btheta}_n) + \bar{\bbc}(\hat{\btheta}_n))$, where $\bar{\bbh}(\hat{\btheta}_n)$ and $\bar{\bbc}(\hat{\btheta}_n)$ are defined in (\ref{eq:HCemp}).
Further, we define the expected gradient matrix of the random vector $\boldsymbol{d}(\btheta|\bU)$ as
$$\nabla D_{\btheta}:=E\left[\partial_{\btheta_k}\boldsymbol{d}_l(\btheta|\bU) \right]_{l=1,\ldots,\frac{p(p+1)}{2},k=1,\ldots,p} \in \bbr^{\frac{p(p+1)}{2}\times p}.$$

Now, under suitable regularity conditions \citep[A1-A10 in][]{White1982}, assuring that $l(\hat{\btheta}_n|\bu_t)$ is a continuous measurable function and its derivatives exist, the following is shown. Given a copula model \citep{HuangProkhorov2011} or a vine copula model \citep{Schepsmeier2013}, the asymptotic covariance matrix of $\sqrt{n}\bar{\boldsymbol{d}}(\hat{\btheta}_n)$ is given by
{\small
\begin{equation*}
	V_{\btheta_0} = \mathbbm{E}\bigg[(d(\btheta_0|\bU)-\nabla D_{\btheta_0}\bbh(\btheta_0)^{-1}\partial_{\btheta_0} l(\btheta_0|\bU)) \big(d(\btheta_0|\bU)-\nabla D_{\btheta_0}\bbh(\btheta_0)^{-1}\partial_{\btheta_0} l(\btheta_0|\bU)\big)^T\bigg].
	\label{eq:varianceMatrix}
\end{equation*}}\noindent
Here, $\hat{\btheta}_n$ is again the maximum likelihood estimate of $\btheta_0$ given $n$ i.i.d. samples.
For details on the estimation of $\nabla D_{\btheta_0}$ and $V_{\btheta_0}$ we refer to \citet{Schepsmeier2013}.

Thus, the test statistic of the \textbf{White} test is
\begin{equation}
	\text{\textbf{White:}}\qquad \calT_n = n\left(\bar{\boldsymbol{d}}(\hat{\btheta}_n)\right)^T\hat{V}_{\hat{\btheta}_n}^{-1} \bar{\boldsymbol{d}}(\hat{\btheta}_n),
	\label{eq:teststatistic}
\end{equation}
where $\hat{V}_{\hat{\btheta}_n}^{-1}$ is the estimated asymptotic variance matrix given $n$ observation.

The test statistic $\calT_n$ follows asymptoticly a $\chi^2$ distributed random variable with degrees-of-freedom $p(p+1)/2$, where $p$ is the number of vine copula parameters. But \citet{Schepsmeier2013} showed that even for small dimensions the asymptotic theory does not hold for relative large data sets, for example $n=1000$ in 5 dimensions. However for bootstrapped p-values the power is quite satisfying, i.e.~the tests has power against false vine copula model specifications. 
Denoting the $1-\alpha$ quantile of the $\chi^2_{p(p+1)/2}$ distribution by $s_{1-\alpha}$ the White test rejects the null hypothesis (\ref{eq:test}) if $\calT_n>s_{1-\alpha}$. In the bootstrapped case the $\chi^2_{p(p+1)/2}$ distribution is replaced by the empirical distribution function of the bootstrapped test statistics. 

\subsection{Rosenblatt's transform tests}
\label{subsec:PIT}

The vine copula GOF test suggested by \citet{Aas_Czado} is based on the multivariate probability integral transform (PIT) of \citet{Rosenblatt1952} applied to copula data $\bu=(\bu_{1}^T,\ldots,\bu_d^T), \bu_i=(u_{1i},\ldots,u_{ni})^T,$ $i=1,\ldots,d$ and a given estimated vine copula model $(\mathcal{V},\mathcal{B}(\mathcal{V}),\hat{\btheta}(\mathcal{B}(\mathcal{V})))$. The general multivariate PIT definition and the explicit algorithm for the R-vine copula model is given in \ref{appendix:rosenblatt}.
The PIT output data $\boldsymbol{y}=(\boldsymbol{y}_1^T,\ldots,\boldsymbol{y}_d^T), y_i=(y_{1i},\ldots,y_{ni})^T, i=1\ldots,d$ is assumed to be i.i.d.~with $y_{it}\sim U[0,1]$ for $t=1,\ldots,n$. Now, a common approach in multivariate GOF testing is dimension reduction. Here the aggregation is performed by
\begin{equation}
	s_{t} := \sum_{i=1}^d \Gamma(y_{ti}),\quad t=\{1,\ldots,n\},
	\label{eq:weightFunction}
\end{equation}
with a weighting function $\Gamma(\cdot)$.
\begin{figure}[htb]
\[
\begin{pmatrix}
u_{11} & \ldots & u_{1d} \\
\vdots & & \vdots \\
u_{n1} & \ldots & u_{nd} \\
\end{pmatrix} \xrightarrow[(PIT)]{Rosenblatt}
\begin{pmatrix}
y_{11} & \ldots & y_{1d} \\
\vdots & & \vdots \\
y_{n1} & \ldots & y_{nd} \\
\end{pmatrix} \xrightarrow[\Gamma(y_{ti})]{Aggregation}
\begin{pmatrix} 
s_1 \\
\vdots \\
s_n \\
\end{pmatrix} \xrightarrow[GOF tests]{univariate}
\]
\caption{Schematic procedure of the PIT based goodness-of-fit tests.}
\label{fig:PIT}
\end{figure}
\noindent
\citet{Breymann2003} suggest as weight function the squared quantile of the standard normal distribution, i.e. $\Gamma(y_{ti}) = \Phi^{-1}(y_{ti})^2$, with $\Phi(\cdot)$ denoting the $N(0,1)$ cdf. Finally, they apply a univariate Anderson-Darling test to the univariate test data $s_t$. The three step procedure is summarized in Figure \ref{fig:PIT}.


\citet{BergBakken2007} point out that the approach of \citet{Breymann2003} has some weaknesses and limitations. The weighting function $\Phi^{-1}(y_{ti})^2$ strongly weights data along the boundaries of the $d$-dimensional unit hypercube.
They suggest a generalization and extension of the PIT approach. First, they propose two new weighting functions for the aggregation in (\ref{eq:weightFunction}): 
\[
	\Gamma(y_{ti})=|y_{ti}-0.5|\quad \text{and}\quad
	\Gamma(y_{ti})=(y_{ti}-0.5)^{\alpha},  \alpha=(2,4,\ldots).
\]

Further, they use the order statistics of the random vector $\boldsymbol{Y}=(Y_1,\ldots,Y_d)$, denoted by $Y_{(1)}\leq Y_{(2)}\leq\ldots\leq Y_{(d)}$ with observed values $y_{(1)}<y_{(2)}\leq\ldots\leq y_{(d)}$. 
The calculation of the order statistics PIT can be simplified by using the fact that
$Y_{(1)}\leq Y_{(2)}\leq\ldots\leq Y_{(d)}$ are i.i.d.~$U(0,1)$ random variables and $\{Y_{(i)}, 1\leq i\leq d\}$ is a Markov chain \cite[Theorem 2.7]{David1981}. Now Theorem 1 of \citet{Deheuvels1984} can be applied and the calculation of the PIT ease to
\begin{equation}
	v_i:=F_{Y_{(i)}|Y_{(i-1)}}(y_{(i)}) = 1-\left(\frac{1-y_{(i)}}{1-y_{(i-1)}}\right)^{d-(i-1)},\quad i=1,\ldots,d, y_{(0)}=0.
	\label{eq:cpit}
\end{equation}

Now, \citet{BergBakken2007} construct the aggregation as the sum of a product of two weighting functions applied to $\boldsymbol{y}$ and $\boldsymbol{v}=(v_1,\ldots,v_d)$, respectively, i.e.
\[
	s_{t} := \sum_{i=1}^d \Gamma_{\boldsymbol{y}}(y_{ti}) \cdot \Gamma_{\boldsymbol{v}}(v_{ti}),\quad t=\{1,\ldots,n\}.
\]
Here $\Gamma_{\boldsymbol{y}}(\cdot)$ and $\Gamma_{\boldsymbol{v}}(\cdot)$ are chosen from the suggested weighting functions including the one of \citet{Breymann2003}. Let $S_t$ be the corresponding random aggregation of $s_t$.
If $\Gamma_{\boldsymbol{y}}(\cdot)=1$ and $\Gamma_{\boldsymbol{v}}(\cdot)=\Phi^{-1}(\cdot)^2$ or vise versa,
the asymptotic distribution of $S_{t}$ follows a $\chi^2_d$ distributed random variable \citep{Breymann2003}. In all other cases the asymptotic distribution of $S_t$ is unknown.

The combinations with $\Gamma_{\boldsymbol{y}}(y_{ti})=|y_{ti}-0.5|$ and $\Gamma_{\boldsymbol{y}}(y_{ti})=(y_{ti}-0.5)^{\alpha}$ for $\alpha=2,4,\ldots$ performed very poorly in the simulation setup considered later. Thus we will not include them in the forthcoming power study. Only the weighting functions listed in Table \ref{tab:SpecificationsOfTheGoodnessOfFitTestA2} will be investigated.
As final test statistics to the test data $s_{t}$ we apply the univariate Cramér-von Mises (CvM) or Kolmogorov-Smirnov (KS) test, as well as the mentioned univariate Anderson-Darling (AD) test. All three test statistics are given in  \ref{appendix:tests} for the convenience of the reader.

\begin{table}[htbp]
	\centering
		\begin{tabular}{lll}
			\toprule
			Short & \multicolumn{2}{c}{Description} \\
			\midrule
			 \textbf{Breymann}	 & $\Gamma_{\boldsymbol{y}}(y_{ti}) = \Phi^{-1}(y_{ti})$ & $\Gamma_{\boldsymbol{v}}(v_{ti}) = 1$ \\
			 \textbf{Berg}			 & $\Gamma_{\boldsymbol{y}}(y_{ti}) = 1$ & $\Gamma_{\boldsymbol{v}}(v_{ti}) = |v_{ti}-0.5|$ \\
			 \textbf{Berg2}			 & $\Gamma_{\boldsymbol{y}}(y_{ti}) = 1$ & $\Gamma_{\boldsymbol{v}}(v_{ti}) = (v_{ti}-0.5)^2$ \\
			\bottomrule
		\end{tabular}
	\caption{Specifications of the PIT based goodness-of-fit tests.}
	\label{tab:SpecificationsOfTheGoodnessOfFitTestA2}
\end{table}

Let $s_{1-\alpha}^{AD}, s_{1-\alpha}^{CvM}$ and $s_{1-\alpha}^{KS}$ denote the $1-\alpha$ quantile of the univariate AD, CvM or KS test statistic, respectively. Then the test rejects the null hypothesis (\ref{eq:test2}) if $W^2_n>s_{1-\alpha}^{AD}, \omega^2>s_{1-\alpha}^{CvM}$ or $D_n>s_{1-\alpha}^{KS}$, respectively.

\subsection{Empirical copula process tests}
\label{subsec:ecp}

A rather different approach is suggested by \citet{GenestRemillard2007} for copula GOF testing. They propose to use the difference of the copula distribution function $C_{\hat{\btheta}_n}(\bu)$ with estimated parameter $\hat{\btheta}_n$ and the empirical copula $\hat{C}_n(\bu)$ (see Equation (\ref{eq:empCopula})) given the copula data $\bu$. This stochastic process is known as the empirical copula process (ECP) and will be used to test (\ref{eq:test2}).
For a vine copula model the copula distribution function $C_{\hat{\btheta}_n}(\bu)$ is not given in closed form. Thus a bootstrapped version has to be used.

Now, the ECP $\hat{C}_n(\bu)-C_{\hat{\btheta}_n}(\bu)$ is utilized in a multivariate Cramér-von Mises (mCvM) or multivariate Kolmogorov-Smirnov (mKS) based test statistic. The multivariate distribution functions $\hat{F}_n(\boldsymbol{y})$ and $F(\boldsymbol{y})$ in Equation (\ref{eq:CvM}) and (\ref{eq:KS}) of \ref{subsec:CvM} are replaced by their (vine) copula equivalents $\hat{C}_n(\bu)$ and $C_{\hat{\btheta}_n}(\bu)$, respectively. Thus we consider
\begin{equation*}
\begin{split}
	&\text{\textbf{ECP-mCvM:}}\quad n \omega^2_{ECP} := n \int_{[0,1]^d} (\hat{C}_n(\bu)-C_{\hat{\btheta}_n}(\bu))^2 d\hat{C}_n(\bu) 
	\quad\text{ and} \\
	&\text{\textbf{ECP-mKS:}}\quad D_{n,ECP} := \sup_{\bu\in[0,1]^d} |\hat{C}_n(\bu)-C_{\hat{\btheta}_n}(\bu)|.
\end{split}
\end{equation*}

To avoid the calculation/approximation of $C_{\hat{\btheta}_n}(\bu)$ \citet{Genest2009} and other authors propose to use the transformed data $\boldsymbol{y}=(y_1,\ldots,y_d)$ of the PIT approach and plug them into the ECP. The idea is to calculate the distance between the empirical copula $\hat{C}_n(\boldsymbol{y})$ of the transformed data $\boldsymbol{y}$ and the independence copula $C_{\bot}(\boldsymbol{y})$.
Thus, the considered multivariate CvM and KS test statistics are
\begin{equation*}
\begin{split}
	&\text{\textbf{ECP2-mCvM:}}\quad n \omega^2_{ECP2} := n \int_{[0,1]^d} (\hat{C}_n(\boldsymbol{y})-C_{\bot}(\boldsymbol{y}))^2 d\hat{C}_n(\boldsymbol{y})
	\quad\text{ and} \\
	&\text{\textbf{ECP2-mKS:}}\quad D_{n,ECP2} := \sup_{\boldsymbol{y}\in[0,1]^d} |\hat{C}_n(\boldsymbol{y})-C_{\bot}(\boldsymbol{y})|,
\end{split}
\end{equation*}
respectively. 
Since neither the mCvM nor the mKS test statistic has a known asymptotic distribution function a parametric bootstrap procedure has to be applied to estimate p-values. Thus a computer intensive double bootstrap procedure has to be implemented. As before the test rejects the null hypothesis (\ref{eq:test2}) if $n\omega^2_{ECP}>s_{1-\alpha}^{mCvM}$ or $D_{n,ECP}>s_{1-\alpha}^{mKS}$, respectively. Here $s_{1-\alpha}^{mCvM}$ and $s_{1-\alpha}^{mKS}$ are the $1-\alpha$ quantiles of the mCvM and mKS test statistic's empirical distribution function, respectively. Similar rejection regions are defined for the ECP2 test statistics.

\section{Power study}
\label{sec:powerstudy}

To investigate the power behavior of the proposed GOF tests and to compare them to each other we conduct several Monte Carlo studies of different dimension. The second property of interest is the ability of the test to maintain the nominal level or size, usually chosen at 5\%.

If a test has the probability of rejection less than or equal to a small number $\alpha\in(0,1)$, called the level of significance, for the hypothesis $H_0$, then such a test is called a $\alpha$-level test. We speak of rejecting $H_0$ at level $\alpha$. Common values for $\alpha$ are 0.05 and 0.01. Since a test of level $\alpha$ is also a test of level $\alpha^{\prime}>\alpha$, the smallest such $\alpha$ is called \textit{size} of the test and is the maximum probability of type I error \cite[p.217]{Bickel2007}. The \textit{power} of a test against the alternative $H_1$ is the probability of rejecting $H_0$ when $H_1$ is true. It is often denoted as $\beta$.

Given an observed test statistic $t_n$ of $\calT_n$ the corresponding  p-value is defined as
\[
	p(t_n):=P(\calT_n\geq t_n).
\] 
Here $\calT_n$ represents one of the test statistics $\calT_n$ (White), $Z_n$ (IR), $W^2_n$ (AD), $n\omega^2$ (CvM or mCvM) and $D_n$ (KS or mKS) introduced in Section \ref{sec:IR} and \ref{sec:gof}.

For a given model $M_1$ consider the random statistic $\calT_n(M_1)$ based on an i.i.d.~sample of size $n$ from model $M_1$ with observed value $t_n(M_1)$. Define the random variable $Z_{M_1}:=p(\calT_n(M_1))$ which takes on values $z_{M_1}=p(t_n(M_1))$ in $(0,1)$. Let $F_{M_1}(\cdot)$ denote the distribution function of $Z_{M_1}$, then $F_{M_1}(\alpha)$ is the \textit{actual size} of the test at level $\alpha$ (\textit{nominal size}). A test maintains its nominal level if $F_{M_1}(\alpha)=\alpha$. As estimates of the p-value and the distribution function we use their empirical versions. Therefore generate $B$ bootstrap realizations of the test statistic $\calT_n(M_1)$, denoted as $t_n^j(M_1), j=1,\ldots,B$, when $n$ observations are drawn from model $M_1$.Then the estimate of $p_{M_1}^j:=p(T_n^j(M_1))$ is given as
\[
	\hat{p}_{M_1}^j:=\hat{p}(t_n^j(M_1)):=\frac{1}{B}\sum_{r=1}^B \boldsymbol{1}_{\{t_n^r(M_1)\geq t_n^j(M_1)\}}.
\]
Further, the estimated size at level $\alpha\in(0,1)$ is defined as
$
	\hat{F}_{M_1}(\alpha):=\frac{1}{B}\sum_{r=1}^B \boldsymbol{1}_{\{\hat{p}_{M_1}^r\leq\alpha\}}.
$
Generating $B$ i.i.d.~data sets of an alternative model $M_2$ in $H_1$ to estimate $F_{M_2}(\alpha)$ by $\hat{F}_{M_2}(\alpha)$ we get the \textit{power} of the test when the alternative $H_1$ holds.

\subsection{General simulation setup}
\label{subsec:testProcedure}

For the general simulation setup we follow the procedure of \citet{Schepsmeier2013}. 
Given a vine copula model $M_1=RV(\mathcal{V}_1,\mathcal{B}_1(\mathcal{V}_1),\btheta_1(\mathcal{B}_1(\mathcal{V}_1)))$ we test for each proposed GOF test if it has suitable power against an alternative vine copula model $M_2=RV(\mathcal{V}_2,\mathcal{B}_2(\mathcal{V}_2),\btheta_2(\mathcal{B}_2(\mathcal{V}_2)))$, where $M_2\neq M_1$, as follows:
\begin{algorithmic}[1]
\STATE Set vine copula model $M_1$.
\STATE Generate a copula data sample of size $n=1000$ from model $M_1$ (pre-run).
\STATE Given the data of the pre-run select and estimate $M_2$ using e.g.~the step-wise selection algorithm of 		\citet{DissmannBrechmannCzadoKurowicka2011}.
\FOR{$r=1,\ldots,B$}
	\STATE Generate copula data $\bu_{M_1}^r=(\bu_{M_1}^{1r},\ldots,\bu_{M_1}^{dr})$ from $M_1$ of size $n$.
	\STATE Estimate $\btheta_1(\mathcal{B}_1(\mathcal{V}_1))$ of model $M_1$ given data $\bu_{M_1}^r$ and denote it by\\ $\hat{\btheta}_1(\mathcal{B}_1(\mathcal{V}_1);\bu_{M_1}^r)$.
	\STATE Calculate test statistic $t_n^r(M_1):=t_n^r(\hat{\btheta}_1(\mathcal{B}_1(\mathcal{V}_1);\bu_{M_1}^r))$ based on data $\bu_{M_1}^r$ assuming the vine copula model $M_1=RV(\mathcal{V}_1,\mathcal{B}_1(\mathcal{V}_1),\hat{\btheta}_1(\mathcal{B}_1(\mathcal{V}_1)))$.
	\STATE Generate copula data $\bu_{M_2}^r=(\bu_{M_2}^{1r},\ldots,\bu_{M_2}^{dr})$ from $M_2$ of size $n$.
	\STATE Estimate $\btheta_1(\mathcal{B}_1(\mathcal{V}_1))$ of model $M_1$ given data $\bu_{M_2}^r$ and denote it by\\ $\hat{\btheta}_1(\mathcal{B}_1(\mathcal{V}_1);\bu_{M_2}^r)$.
	\STATE Calculate test statistic $t_n^r(M_2):=t_n^r(\hat{\btheta}_1(\mathcal{B}_1(\mathcal{V}_1);\bu_{M_2}^r))$ based on data $\bu_{M_2}^r$ assuming vine copula model $M_1$.
\ENDFOR
\STATE Estimate p-values $p^j_{M_1}$ and $p^j_{M_2}$ by \\
 $\hat{p}^j_{M_1} = \hat{p}(\calt^j_n(M_1)) := \frac{1}{B}\sum_{r=1}^B \boldsymbol{1}_{\{\calt^r_n(M_1)\geq \calt^j_n(M_1)\}}$\quad and\\
	$\hat{p}^j_{M_2} = \hat{p}(\calt^j_n(M_2)) := \frac{1}{B}\sum_{r=1}^B \boldsymbol{1}_{\{\calt^r_n(M_2)\geq \calt^j_n(M_2)\}},$\quad respectively, for $j=1,\ldots,B$.
\STATE Estimate the distribution function of $Z_{M_1}$ and $Z_{M_2}$ by \\ 
$
		\hat{F}_{M_1}(\alpha) := \frac{1}{B}\sum_{r=1}^B \boldsymbol{1}_{\{\hat{p}^r_{M_1}\leq \alpha\}} \quad\text{and}\quad
		\hat{F}_{M_2}(\alpha) := \frac{1}{B}\sum_{r=1}^B \boldsymbol{1}_{\{\hat{p}^r_{M_2}\leq \alpha\}},
$\\
	respectively, giving size and power.
\end{algorithmic}
\

In all of the forthcoming simulation studies we used $B=2500$ replications and the number of observations were chosen to be $n=500, n=750, n=1000$ or $n=2000$. As model dimension we chose $d=5$ and $d=8$ and the critical level $\alpha$ is $0.05$. Possible pair-copula families in the investigated vine copula models are the elliptical Gauss and Student's t-copula, the Archimedean Clayton, Gumbel, Frank and Joe copula, and their rotated versions.
Further, all calculations are performed using the statistical software R\footnote{R Development Core Team (2012). R: A language and environment for statistical computing. R Foundation for Statistical Computing, Vienna, Austria. ISBN 3-900051-07-0, URL http://www.R-project.org/.} and the R-package \texttt{VineCopula} of \citet{VineCopula}.

\subsection{Test specification}
\label{subsec:testSpecification}

In our power study we investigate the size and power behavior of the proposed GOF tests given an R-vine as true model ($M_1$) with respect to the alternatives
\begin{itemize}
	\item multivariate Gauss copula,
	\item C-vine copula and
	\item D-vine copula.
\end{itemize}
Details on the R-vine structure (Figure \ref{fig:5dimRvine}), the chosen pair-copula families and copula parameters for $d=5$ are given in Table \ref{tab:5dimRvine} in \ref{appendix:powerstudy}. For the 8 dimensional example we refer to Table \ref{tab:8dimRvine} of \ref{appendix:powerstudy}. 

The estimated C- and D-vine structures ($\hat{\mathcal{V}}_C$ and $\hat{\mathcal{V}}_D$) are given in Equation (\ref{eq:5dimCvine}) and (\ref{eq:5dimDvine}) of \ref{appendix:powerstudy}, respectively. The structure selection of the D-vine copula is facilitated by solving a traveling salesman algorithm while the root order of the C-vine model follows the heuristic of \citet{CzadoSchepsmeierMin2011}. The assignment of the pair-copula families in the C- and D-vine uses AIC as suggested and validated in \citet{BrechmannDA}. The last alternative copula model is the multivariate Gauss copula, which can be formulated as a vine copula as well \citep[see][]{Czado}. In the Gaussian case the conditional correlation parameters, which form the pair-copula parameters, are equal to the partial correlation parameters. They can be calculated recursively using the entries of the multivariate Gauss copula variance-covariance matrix.

Although all three stated alternatives have different vine structures and pair-copula families we do not know which vine copula model is ``closer'' to the true R-vine model. A often proposed approach for model comparison is the \citet{KullbackLeibler1951} information criterion (KLIC). It measures the distance between a true unknown distribution and a specified, but estimated model. In the following definition we follow \citet{Vuong}. Let $c_0(\cdot)$ be the true (vine) copula density function of a $d$-dimensional random vector $\bU$. Further, $E_0$ denotes the expected value with respect to this true distribution. The estimated (vine) copula density of $\bU$ is denoted as $c(\hat{\btheta}_n|\bU)$, where $\hat{\btheta}_n$ is the estimated model parameter (vector) given $n$ samples of $\bU$. Then, the KLIC between $c_0$ and $c$ is defined as
\[
	KLIC(c_0,c):=\int_{(0,1)^d} c_0(\bu)\ln\left(\frac{c_0(\bu)}{c(\hat{\btheta}_n|\bu)}\right) d\bu = E_0[\ln c_0(\bU)] - E_0[\ln c(\hat{\btheta}_n|\bU)].
\]
The model with the smallest KLIC is ``closest'' to the true model. In the plots of the following power study we ordered the alternatives on the x-axis by their KLIC as listed in Table \ref{tab:KullbackLeiblerDistances}, e.g.~for $d=5$ we have the order D-vine, C-vine, Gauss.

The approximation of the multidimensional integral is facilitated by Monte Carlo or a numerical integration based on the R package \texttt{cubature} \citep{cubature}. In the numerical integration copula data, i.e. $\bu\in(0,1)^d$, or standard normal transformed data, i.e. $\bx=\Phi(\bu)\in\bbr^d$, are used.
We see that it is quite challenging to estimate the KLIC distance in high dimensions.

\begin{table}[htbp]
	\centering
	\begin{minipage}{\linewidth}
      \renewcommand{\footnoterule}{}
      \renewcommand{\thefootnote}{\alph{footnote}} 
		\begin{tabular}{rllll}
		\toprule
			d & method & C-vine & D-vine & Gauss \\
		\midrule
		5 & Monte Carlo & 0.65 & 0.64 & 0.72 \\
			& numerical integration based on copula margins & 0.62\footnotemark[1] & 0.45\footnotemark[1] & 0.71\footnotemark[1] \\
			& numerical integration based on normal margins & 0.48\footnotemark[1] & 0.51\footnotemark[1] & 0.50\footnotemark[1] \\
		8 & Monte Carlo & 1.66 & 0.13 & 0.73 \\
			& numerical integration based on copula margins & 1.46\footnotemark[2] & 1.29\footnotemark[2] & 1.91\footnotemark[2] \\
			& numerical integration based on normal margins & 2.15\footnotemark[3] & 3.20\footnotemark[3] & 2.14\footnotemark[3] \\
		\bottomrule	
		\end{tabular}
	\caption{Kullback-Leibler distances of the proposed vine copula models with respect to the true R-vine copula model (${}^{a}$estimated relative error $<0.01$, ${}^{b}$estimated relative error $\approx 1.4$, ${}^{c}$estimated relative error $\approx 3.5$). }
	\label{tab:KullbackLeiblerDistances}
	\end{minipage} 
\end{table}

\subsection{Results}
\label{subsec:numericalResults}

\begin{figure}[htbp]
	\centering
	\includegraphics[width=0.9\textwidth]{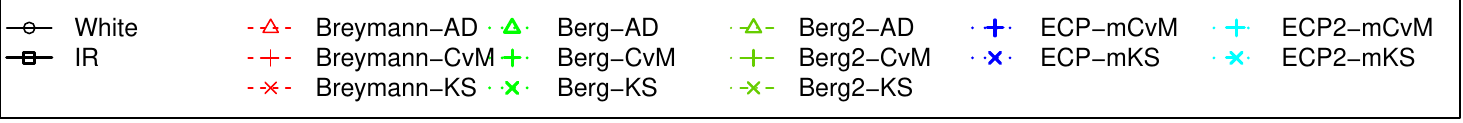}
	\subfigure[$d=5, n=500$]{\includegraphics[width=0.48\textwidth]{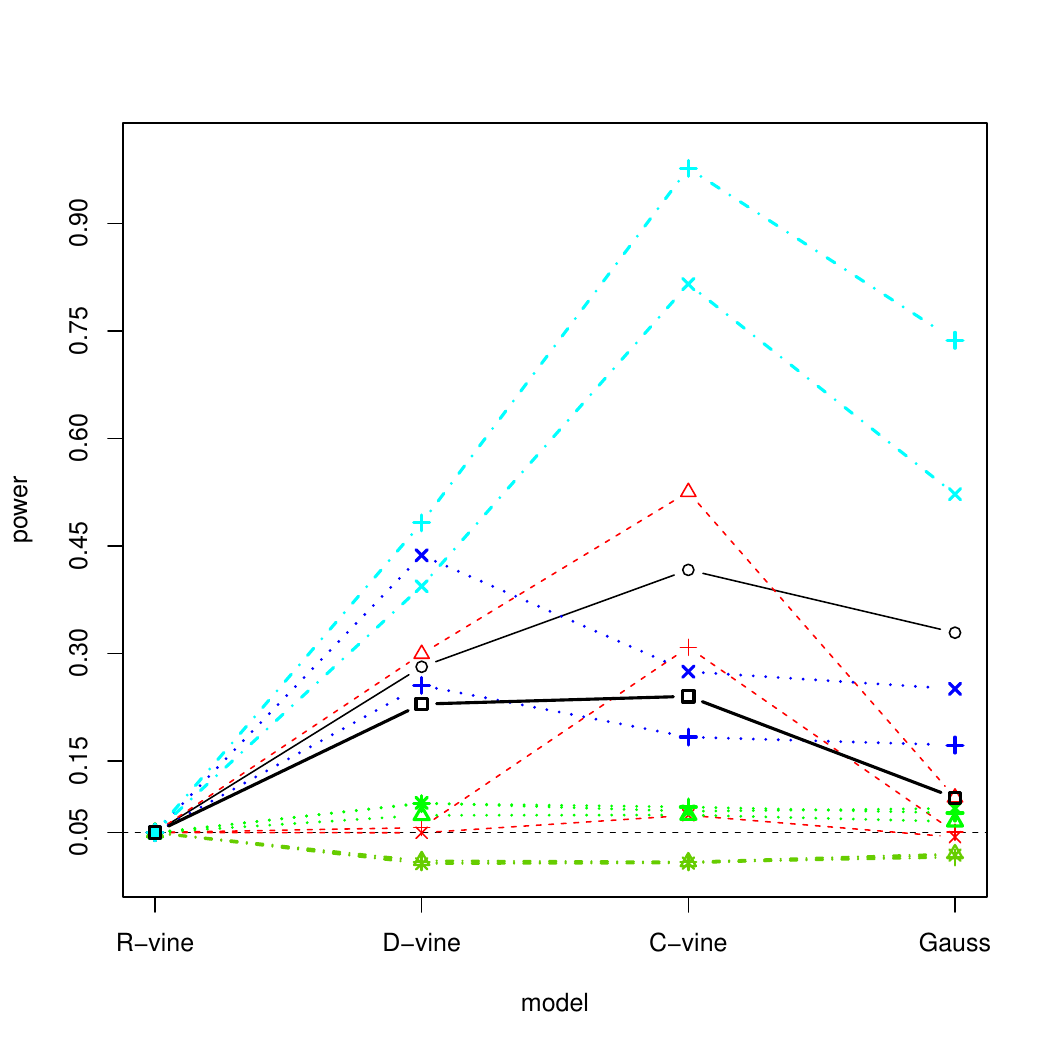}}
	\subfigure[$d=5, n=750$]{\includegraphics[width=0.48\textwidth]{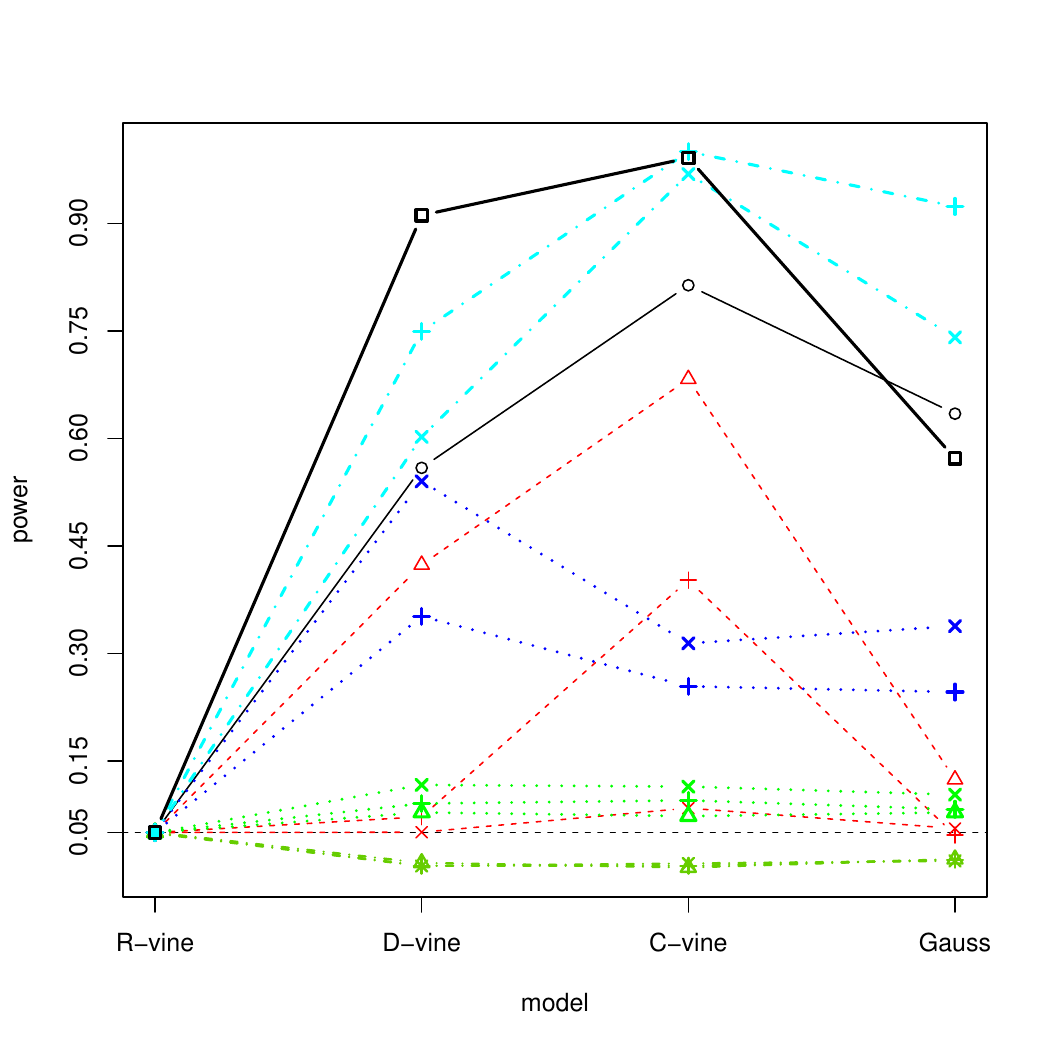}}\\
	\subfigure[$d=5, n=1000$]{\includegraphics[width=0.48\textwidth]{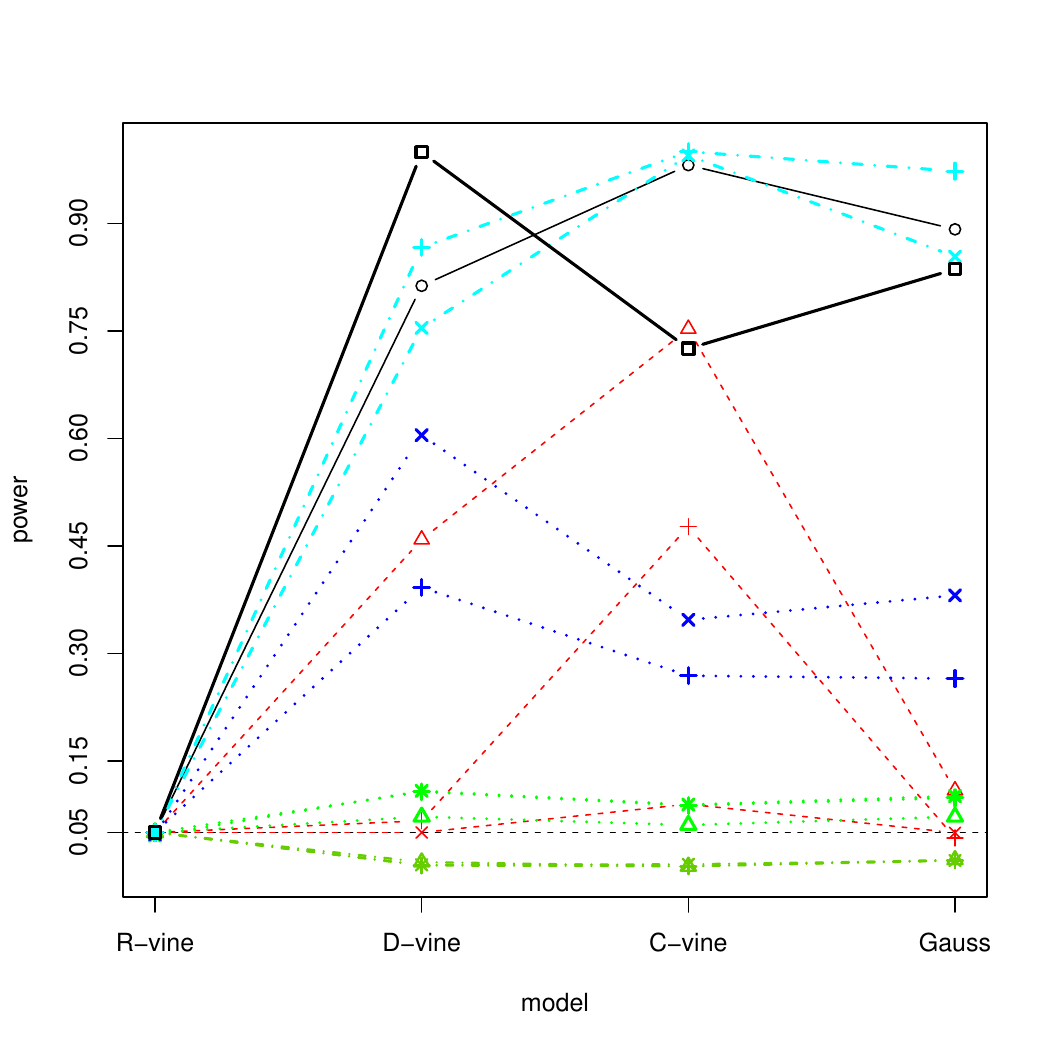}}
	\subfigure[$d=5, n=2000$]{\includegraphics[width=0.48\textwidth]{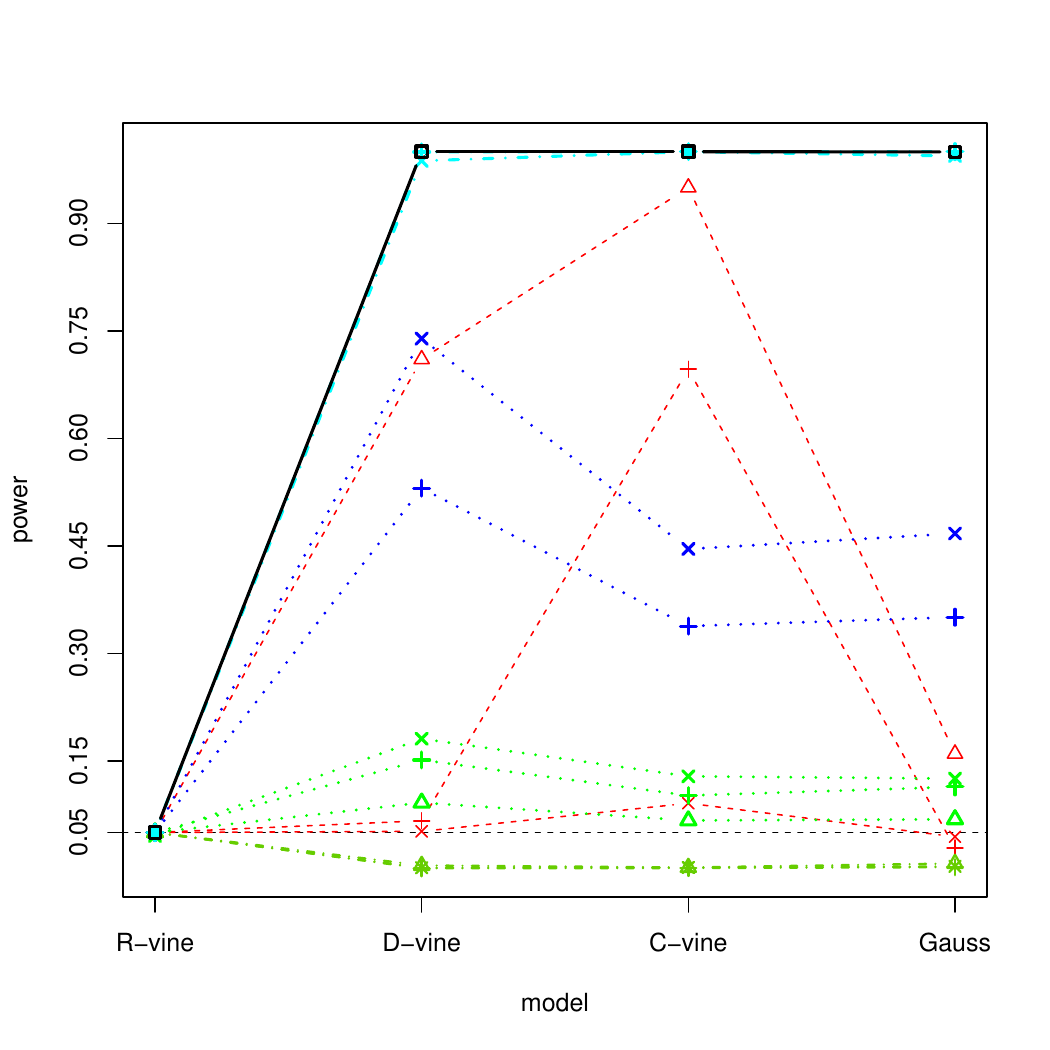}}
	\caption{Power comparison of the proposed goodness-of-fit tests in 5 dimensions with different number of sample points. The alternatives are ordered on the x-axis by the rank of their KLIC value with respect to the true R-vine.}
	\label{fig:powerPlots}
\end{figure}

\begin{figure}[htbp]
	\centering
	\includegraphics[width=0.9\textwidth]{legend.pdf}
	\subfigure[$d=8, n=500$]{\includegraphics[width=0.48\textwidth]{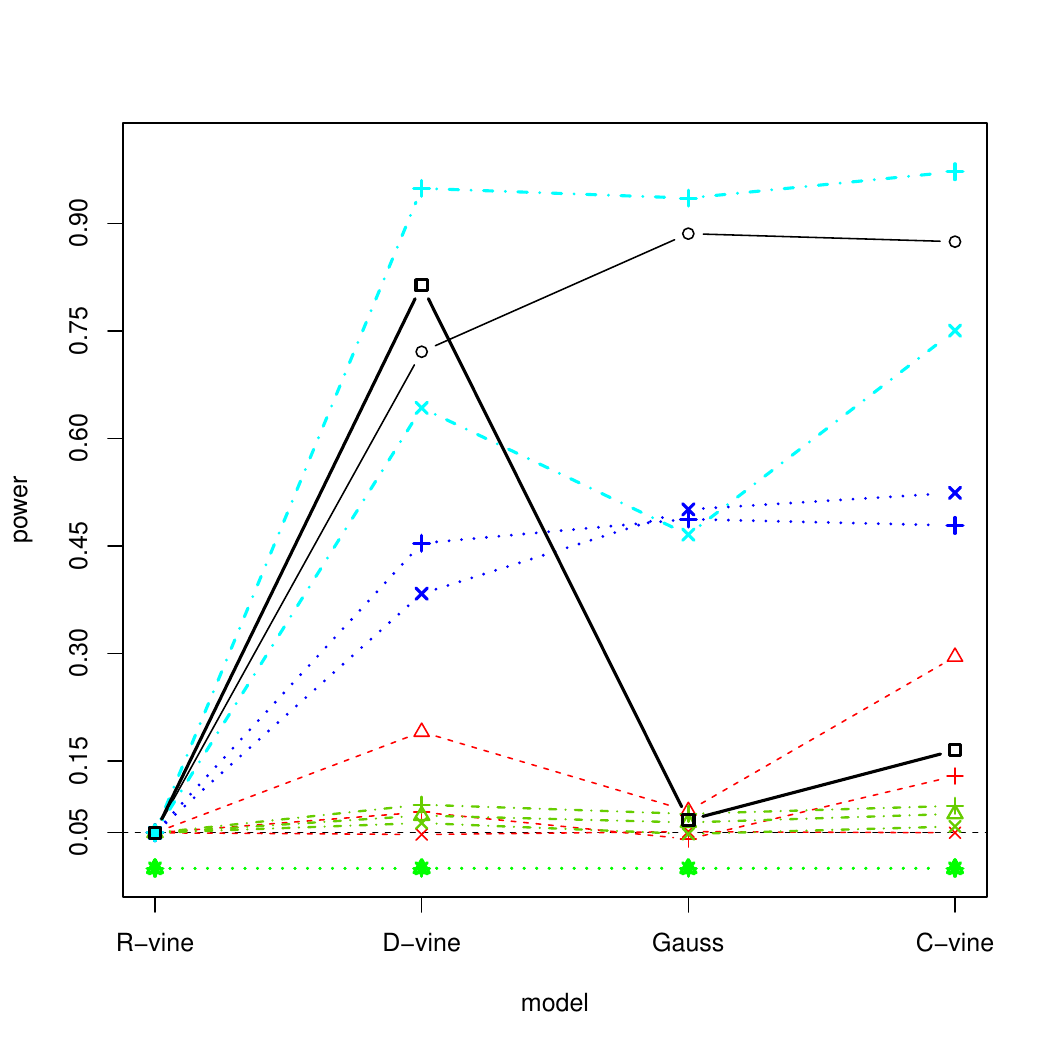}}
	\subfigure[$d=8, n=750$]{\includegraphics[width=0.48\textwidth]{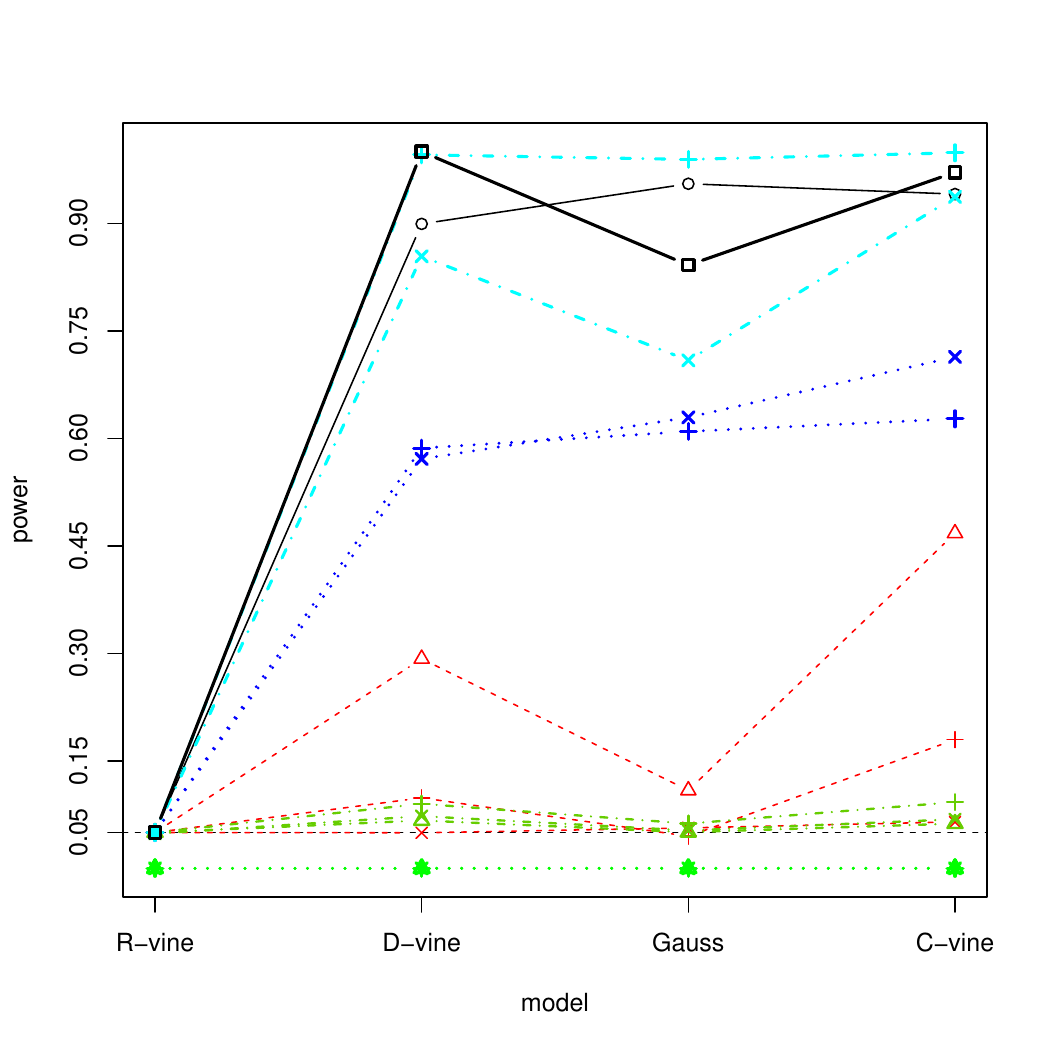}}
	\subfigure[$d=8, n=1000$]{\includegraphics[width=0.48\textwidth]{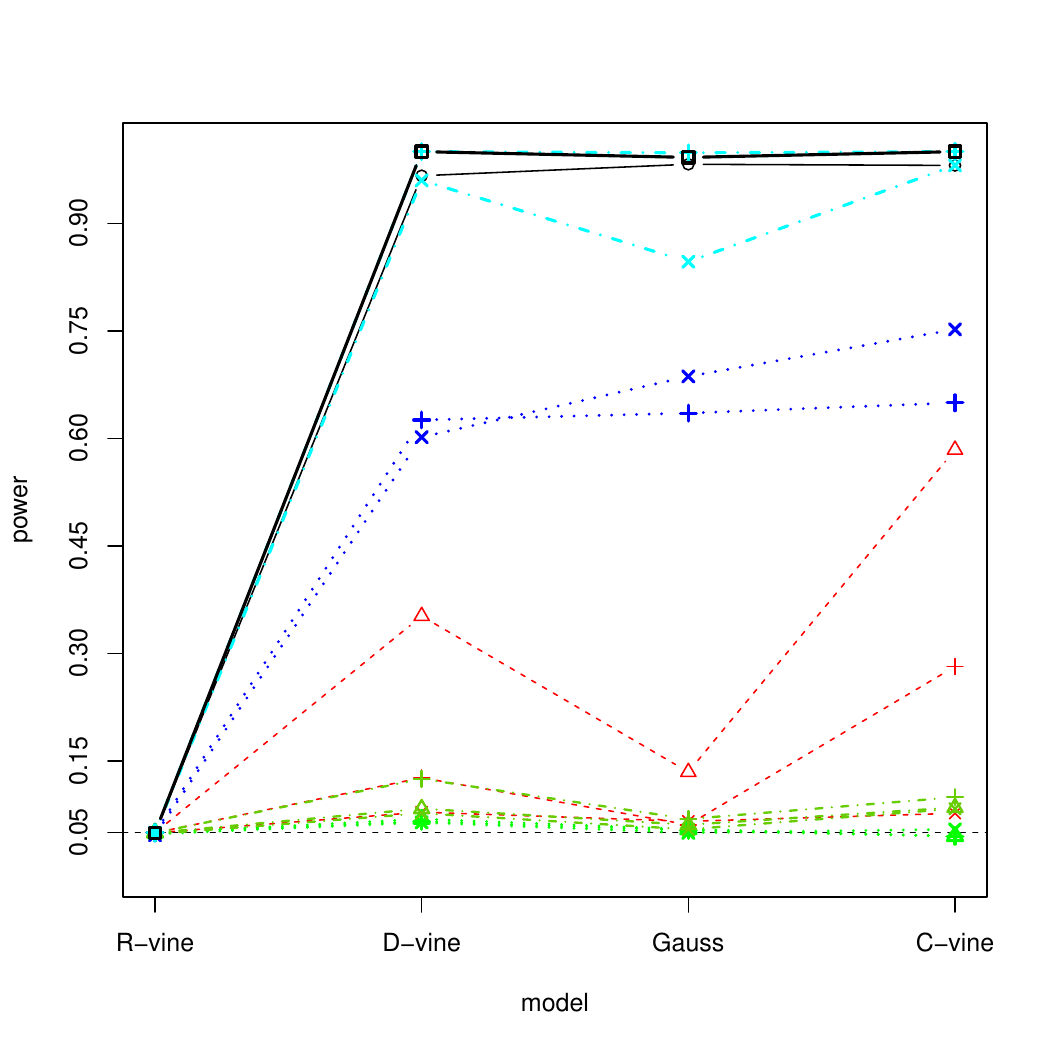}}
	\subfigure[$d=8, n=2000$]{\includegraphics[width=0.48\textwidth]{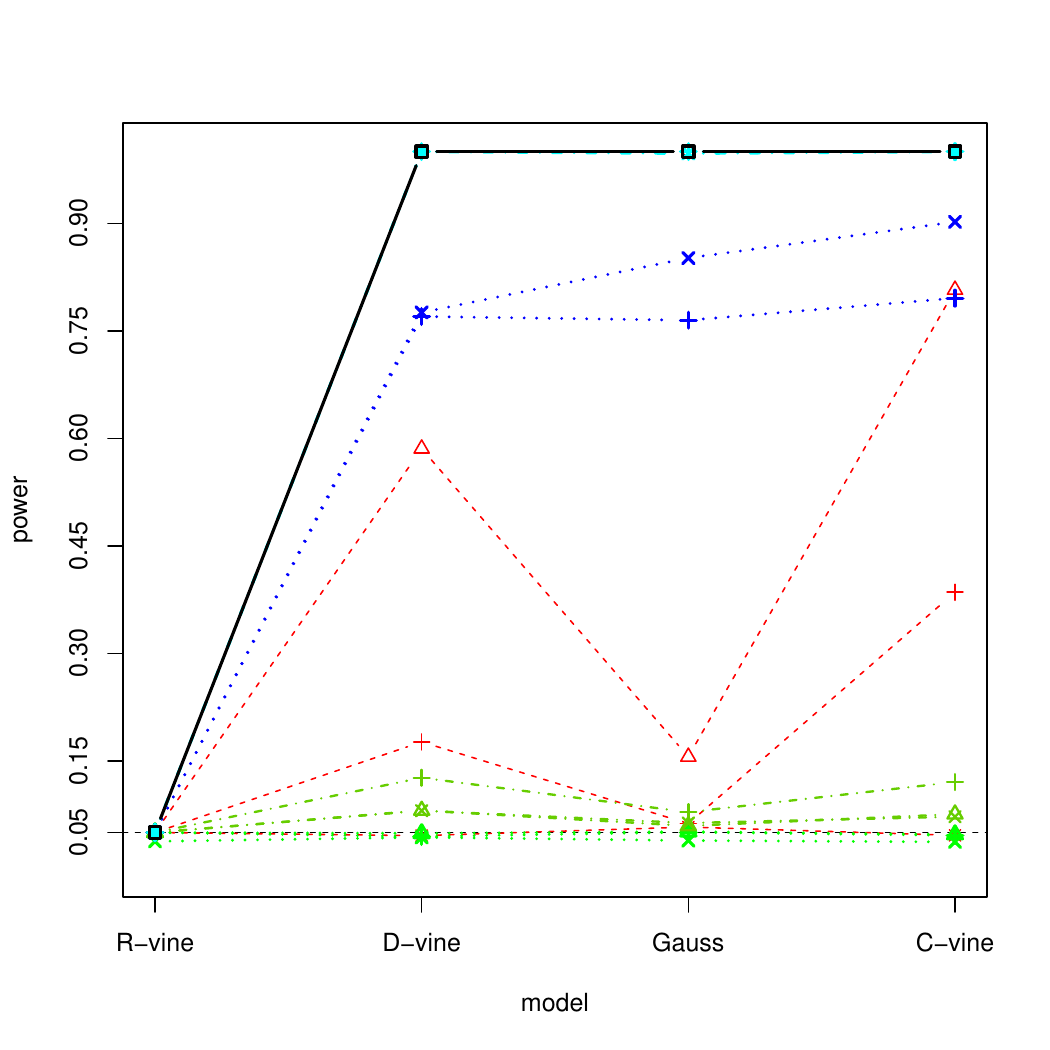}}
	\caption{Power comparison of the proposed goodness-of-fit tests in 8 dimensions with different number of sample points. The alternatives are ordered on the x-axis by the rank of their KLIC value with respect to the true R-vine.}
	\label{fig:powerPlots8d}
\end{figure}

Since all proposed GOF tests have either no asymptotic distribution at all or face substantial numerical problems estimating the asymptotic variance or have shown to have low power in small samples, we only investigate the bootstrapped version of the tests. In the Figures \ref{fig:powerPlots} and \ref{fig:powerPlots8d} we illustrate the estimated power of all 15 proposed GOF tests for $d=5$ and $d=8$, respectively. On the x-axis we have the R-vine as true model and the three alternatives ordered by their KLIC. For the true model the actual size is plotted. A horizontal black dashed line indicates the 5\% $\alpha$-level.\\

\textbf{Size:}
All proposed GOF tests maintain their given size independently of the number of sample points for $d=5$.
In the 8-dimensional case the GOF tests based on the Berg approaches do not maintain their nominal size in case of $n=500$ and $n=750$. All other GOF tests do hold the $5\%$ level and thus control the type I error. 

\textbf{Sample size effects on power:}
	We have increasing power with increasing sample size for the White, IR, ECP, ECP2 and Breymann (in combination with the AD test statistic) GOF test.
		The tests based on Berg and Berg2 have no or very low power independently of the number of observations. This is also true for the Breymann GOF test in combination with the univariate CvM and KS test statistics.
		In eight dimensions the number of sample points are important for the IR test since the tests has very small power considering only 500 data points. In five dimensions the effect is not that eye-catching but can be found too.
		Almost independent from the the number of sample points is the ordering of the test by their power. In all test scenarios the ECP2 test with mCvM test statistic outperforms the others, followed by the IR test, the test based on White and the ECP2 test based on the mKS test statistic. The next GOF tests are the tests based on the ECP and the Breymann transformation with AD test statistic.

\textbf{Dimension effect on the power:}
		The power of the top four GOF tests (IR, White, ECP and ECP2) are almost independent of the dimension. Only in the case of $n=500$ sample points a clearly increase of power can be observed from $d=5$ to $d=8$ dimensions.
		For the weaker tests the reverse is true. With increasing dimension the Breymann GOF test decreases in power. The Berg and Berg2 tests are independently of the dimension. 

\textbf{Effect of alternatives on the power:}
	The results with respect to the KLIC are two-fold.
		For $d=5$ the power increases with increasing KLIC for the most GOF tests except for the Gauss copula in $H_1$. 
		For $d=8$ it is again the multivariate Gauss copula which is out of line for many of the tests. The exceptions are the ECP tests.
		For $n\geq 1000$ the power of the four ``good'' tests mentioned before increases with KLIC. Some of them have even a power of 100\%.
The Breymann test is conspicuous, since the test 
is working quite well for the C- and D-vine alternative but is relatively poor for the multivariate Gaussian copula independent of the dimension or sample size. While the Breymann tests have much lower power than the four best GOF tests, they still have power to distinguish between the null and alternative models. 

\textbf{Effects of the test functionals on power:}
For ECP, ECP2 or Breymann tests it appears that CvM based test statistics are more powerful than the KS type test statistics. This is in line with \citet{Genest2009} for bivariate copula GOF tests.

The poor performance of the Breymann, Berg and Berg2 approach was also recognized in the comparison studies of \citet{Genest2009} in the bivariate case and in \citet{Berg} for copulas of dimension 2, 4 and 8. The analyzed copulas in \citet{Berg} were the Gauss, Student's t, Clayton, Gumbel and Frank copula. But there the test statistics maintained their nominal level and had some explanatory power.

The bootstrapped p-values or power values stabilize fast for increasing bootstrap replications, for all GOF tests. This happens for 1000-1500 replications, irrespective of sample size or alternative. In many cases the stabilization is even faster.

Beside these last points, no clear hierarchy among the best performing proposed test statistics is recognizable. But some tests perform rather well while others do not even maintain their nominal level. In particular, our new IR test performs quite well in terms of power against false alternatives.

Of cause the computation time for the different proposed GOF tests is also a point of interest for practical applications. Therefore, in Table \ref{tab:Summary} the computation times in seconds for the different methods run on a Intel(R) Core(TM) i5-2450M CPU @ 2.50GHz computer for $n=1000$ are given alongside with a summary of our findings. The computing time of the information matrix based methods White and IR are clearly higher than the other test statistics. Given the complex calculation of the R-vine gradient and Hessian matrix \citep[see][]{StoeberSchepsmeier2012} this is not very surprising.



\begin{sidewaystable}[htbp]
	\centering
		\begin{tabular}{lccccccc}
		& & & & & & & \\
			\toprule
			& White & Breymann & Berg & Berg2 & ECP & ECP2 & IR \\
			\midrule
			\textbf{main idea} & \footnotesize$\bbh(\btheta)+\bbc(\btheta)=0$ & \multicolumn{3}{c}{PIT+Aggregation+uniform test} 
			& \footnotesize$\hat{C}_n-C_{\hat{\btheta}_n}$ & PIT+ECP &  \footnotesize$-\bbh(\btheta)^{-1}\bbc(\btheta)=I_p$ \\
			\addlinespace
			\textbf{hold nominal level} & + & + & $-$ & $-$ & + & + & + \\
			\addlinespace
			\textbf{power against}  & + & 0 & $-$ & $-$ & 0 & + & + \\
			\textbf{\ \ alternatives} & & (partly) & & & (partly) & & \\
			\addlinespace
			\textbf{consistentency} & + & $-$ & $-$ & $-$ & $-$ & + & + \\
			\addlinespace
			\textbf{asymptotic} & + & 0 & $-$ & $-$ & $-$ & $-$ & + \\
			\textbf{\ \ distribution} & (only for high n) & (proved to be  & & & & & (not tested) \\
			& & incorrect) & & & & & \\
			\addlinespace
			\textbf{complexity} & 0 & $-$ & $-$ & $-$ & + & 0 & + \\
			& (difficult cov. & \multicolumn{3}{c}{(3 step procedure)} & & (2 step & \\
			& matrix) & & & & & procedure) & \\
			\textbf{computation time} & & & & & & \\
			$d=5$ & 3.41 & 0.06 & 0.06 & 0.07 & 0.08 & 0.06 & 1.66 \\
			$d=8$ & 58.79 & 0.14 & 0.14 & 0.18 & 0.17 & 0.13 & 30.30 \\
			\bottomrule
		\end{tabular}
	\caption{Overview of the performance of the proposed GOF tests.}
	\label{tab:Summary}
\end{sidewaystable}


\section{Application}
\label{sec:application}

As application we consider a financial data set of four indices and their corresponding volatility indices, namely the German DAX and VDAX-NEW, the European EuroSTOXX50 and VSTOXX, the US S\&P500 and VIX, and the Swiss SMI and VSMI. The daily data cover the time horizon of the current financial crisis starting at August, 9th, 2007 when a sharp increase of inter bank interest rates was noticed, until April 30th, 2013, resulting in 1405 data points. For each marginal time series we calculated the log-returns and modeled them with an AR(1)-GARCH(1,1) model using Student's t innovations. The resulting standardized residuals are transformed using the non-parametric rank transformation \citep[see][]{Genest1995} to obtain $[0,1]^8$ copula data.

The contour and pair plots in Figure \ref{fig:ContourCrisis} reveal the expected elliptical positive dependence behavior among the indices and among the volatility indices. But between the indices and the volatilities a negative dependence can be observed. Furthermore, a slight asymmetric tail dependence is recognizable.

\begin{figure}[ht!]
	\centering
	\vspace{-1cm}
		\includegraphics[width=0.9\textwidth]{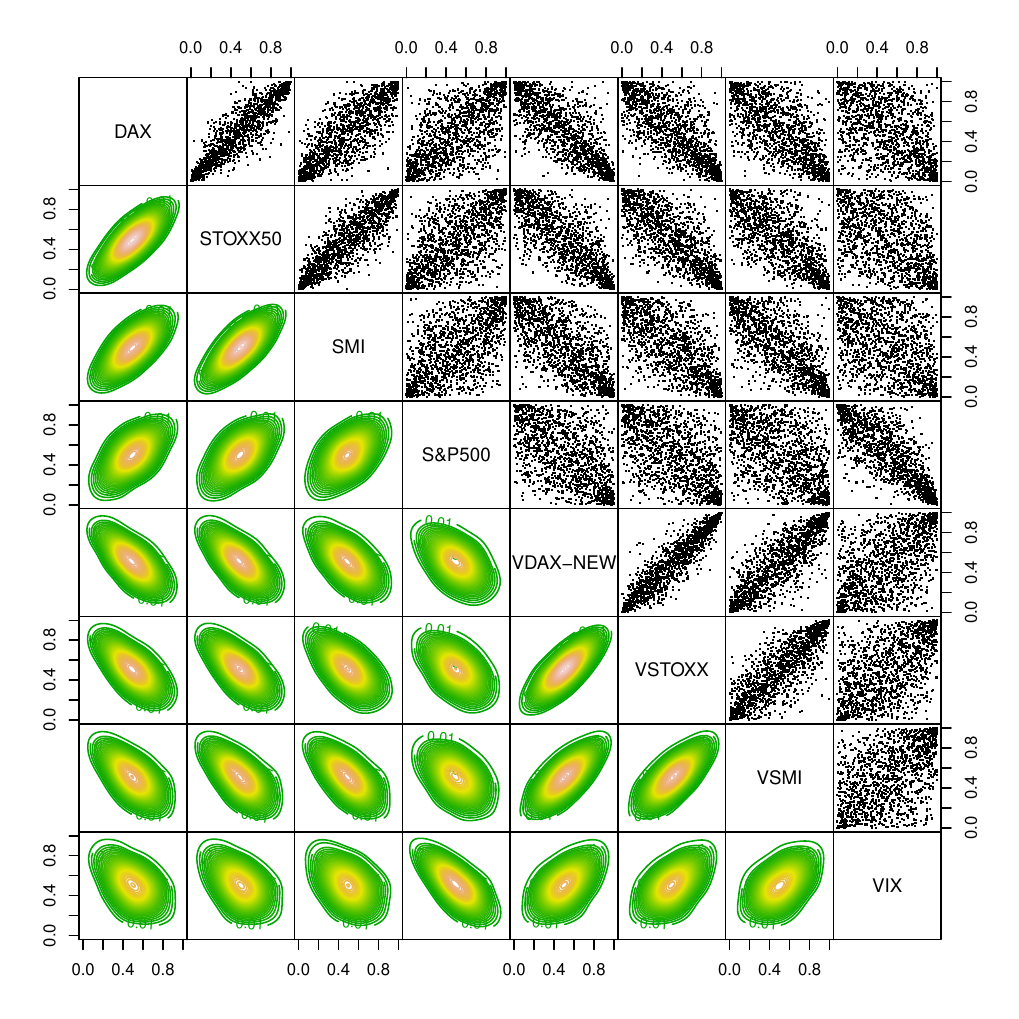}
	\caption{Lower left: Contour plots with standard normal margins; upper right: pairs plots of the transformed data set.}
	\label{fig:ContourCrisis}
\end{figure}

To model the dependence structure we investigated four models. In particular, an R-vine copula model, selected using the maximum spanning tree algorithm by \citet{DissmannBrechmannCzadoKurowicka2011}, a C-vine copula, selected by the heuristic proposed by \citet{CzadoSchepsmeierMin2011}, a D-vine copula, selected using a traveling sales man algorithm, and a multivariate Gaussian copula. The corresponding first trees of the vine models are illustrated in Figure \ref{fig:Tree1Crisis}. For the R-vine copula as well as in the D-vine copula we can see that the indices and the volatilities cluster except for the US ones. The C-vine copula is too restrictive to recognize such groupings. Another interesting point is that the first tree structure of the R-vine is very close to the first tree structure of the D-vine. If we delete the edge ``DAX-VDAX-NEW'' and add a new edge ``VSMI-SMI'' in the R-vine we get the D-vine tree structure. Further, we see evidence of asymmetric tail dependence since (rotated) Gumbel copulas are selected.

\begin{figure}[ht!]
	\centering
		\includegraphics[width=1.00\textwidth]{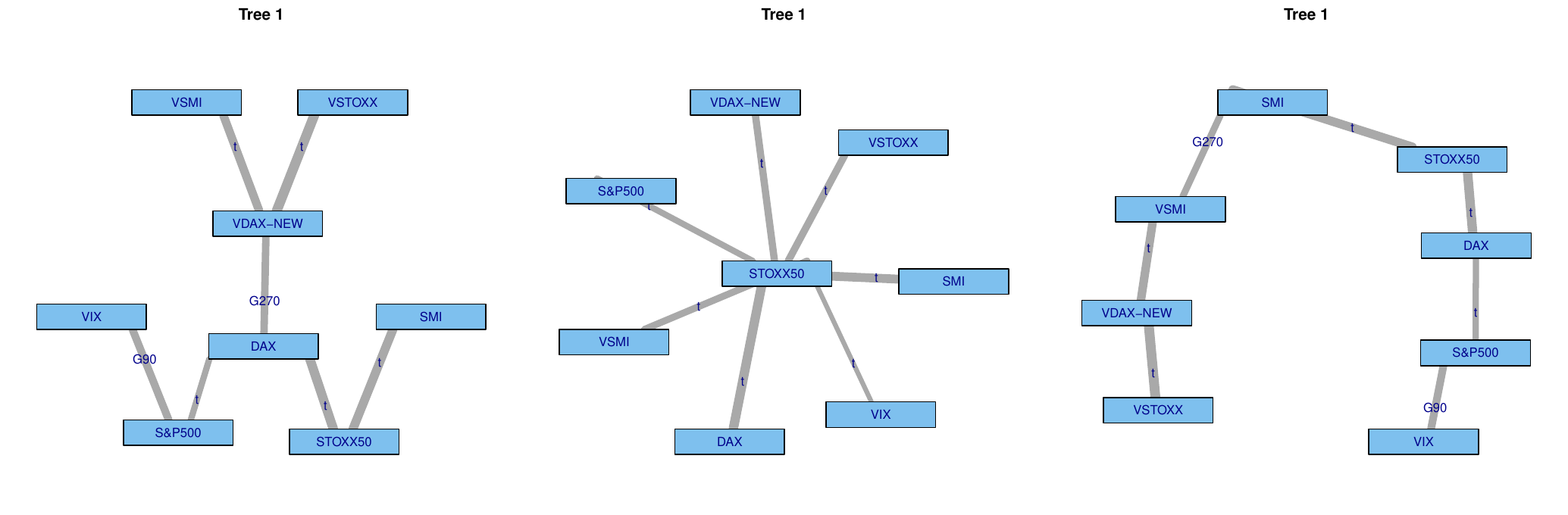}
	\caption{First tree structure of the selected R-vine (left), C-vine (center) and D-vine (right). The edge label denote the corresponding pair-copula family ($t\hat{=}$Student's t; $G90, G270\hat{=}$rotated Gumbel).}
	\label{fig:Tree1Crisis}
\end{figure}

Performing a parametric bootstrap with $B=2500$ most of the good performing proposed GOF tests, namely White, IR, ECP (with CvM) and ECP2 (with CvM), confirm that a vine copula model can not be rejected at a 5\% significance level (see Table \ref{tab:application}). Only the ECP2 approach returns a p-value of 0.01 below the chosen significance level of 0.05 for the estimated C-vine copula, and the White based test a pvalue $<0.01$ for the estimated R-vine copula model. The multivariate Gauss copula is rejected by the White, the IR and ECP2 GOF test, while the ECP based test returns a p-value of 0.6. In 3 of 4 GOF tests the highest returned p-value is for the D-vine copula. But note that the size of the p-value or the ordering of the p-values do not give an ordering of the considered models.

As in the simulation study the GOF tests differ in their rejection decision and several GOF tests are needed to get a better picture of the better fitting model. The discrimination between the estimated vine copula models is even harder than in the power study. The MC-estimated KLIC of the R-vine to the C-vine is only 0.15, while the KLIC of the R-vine to the D-vine is even smaller (0.11). Even the multivariate Gauss copula has an estimated small KLIC with 0.31. Additional simulation studies based on the estimated vine copula models for $n=1000$ show that the simulated power is quite small for all proposed GOF tests. 

In terms of log-likelihood the D-vine is also the best fitting vine copula to the data unless the R-vine has a better AIC and BIC. The significant smaller number of parameters favors the R-vine compared to the D- or C-vine.

The economical interpretation of these findings is, that the assumption of multivariate Gaussian distributed random vectors is not fulfilled in times of financial and economic crises. More flexible models are needed to capture the asymmetric behaviors and tail dependencies. R-vines are able to model these properties as already shown in \citet{BrechmannEuroStoxx2012} or  \citet{StoeberSchepsmeier2012}.

%

\begin{table}[ht]
	\centering
	{\footnotesize
		\begin{tabular}{lrrrrrrrrrr}
			\toprule
			& log-lik & \#par & AIC & BIC  & White & \multicolumn{2}{c}{ECP} & \multicolumn{2}{c}{ECP2} & IR \\
			\cmidrule{7-10}
			& & & & & & CvM & KS & CvM & KS & \\
			\midrule
			R-vine & 7652 & 33 & \textbf{-15238} & \textbf{-15065}  & 0.002 & 0.18 & 0.98 & 0.30 & 0.67 & 0.75 \\
			C-vine & 7585 & 42 & -15086 & -14865 & 0.14 & 0.51 & 0.36 & 0.01 & $<0.01$ & 0.74 \\
			D-vine & \textbf{7654} & 41 & -15226 & -15011 & 0.41 & 0.82 & 0.24 & 0.55 & 0.67 & 0.52 \\
			Gauss  & 7320 & 28 & -14584 & -14445 & $<0.01$ & 0.60 & 0.28 & $<0.01$ & $<0.01$ & $<0.01$ \\
			\bottomrule
		\end{tabular}}
	\caption{Likelihood based validation quantities and bootstrapped p-values of the White, ECP, ECP2 and IR goodness-of-fit test for the 4 considered (vine) copula models}
	\label{tab:application}
\end{table}

\section{Discussion}
\label{sec:discussion}

We introduced a new GOF test for regular vine copula models based on the information matrix ratio. The calculation of the test statistic as well as its asymptotic distribution function showed up to be challenging. But good empirical approximations have been found as shown in an extensive power study. The study revealed good performance of the test in terms of power against false alternatives given simulated p-values. Given sufficient data points the test is even consistent. Furthermore, the new GOF test maintained always its nominal level, controlling the type I error, independently of sample size, dimension or alternative.

Since only \citet{Schepsmeier2013} investigated a GOF test for R-vines so far, further GOF tests extended from the (bivariate) copula case are introduced to facilitate a wider comparison. In particular, 14 other GOF tests were explained and compared in a multi-dimensional setting. The small sample performance for size and power were analyzed for GOF tests based on the difference of the Bartlett identity, the empirical copula process and the multivariate PIT. This paper gives the first comparison study and review of vine copula GOF tests. The new IR test as well as the White test introduced by \citet{Schepsmeier2013} and the ECP2 based tests of \citet{Genest2009} performed very well. They outperformed the tests based on the multivariate PIT. The PIT based tests revealed little power against the considered alternatives. In particular, the new IR GOF test performed best in several cases. Thus, the proposed GOF tests enable statisticians to conduct efficient model diagnostics using hypothesis tests in high dimensional settings.

Of cause further GOF tests already known for copulas can be extended to the vine copula case. But most of them will have crucial problems in higher dimensions. For example the likelihood ratio based GOF test or the Chi-squared type GOF test, both introduced for copulas in \citet{DobricSchmid2005}, have to partition the unit hypercube. This will probably result in long computation time in high dimensions as well as the need of sufficient large number of observations. The Kendall's process based GOF tests suggested by \citet{BergAas2009} need like the ECP based GOF tests a double bootstrap procedure since the Kendall's process is not trackable for the vine copula. But this approach revealed good results in the comparison study of \citet{Genest2009} for bivariate copulas. Further suggestions for copula GOF tests are for example presented in \citet{Fermanian2012}.

A very interesting hybrid approach was suggested by \citet{PeterOstap2013}. Since no GOF test outperforms in all cases a hybrid test is suggested. Given $m$ test statistics $t_n^{(i)}$ with sample size $n$ and controlling type I error for any given significance level $\alpha$ under the null hypothesis the hybrid p-value is defined as
\[
	p_n^{hybrid} := m*\min\{p_n^{(1)},\ldots,p_n^{(m)}\}.
\]
Here $p_n^{(i)}, i=1,\ldots, m$ denote the p-values of the test statistics $t_n^{(i)}$. They showed that the power function is bounded from below and if there is at least one test which is consistent, then the hybrid test is consistent. An extension to the vine copula case would be highly welcomed.

By testing the validity of the null hypothesis $H_0: C\in \calc_0$, where $C$ denotes the (vine) copula distribution function and $\calc_0$ is a class of parametric copulas one has to take the margins into account. 
As pointed out by \citet{Genest2009} the marginal distribution functions $F_1,\ldots,F_d$ of the random variables $X_1,\ldots,X_d$ can be considered as nuisance parameters. So far we always considered known margins. Thus an extension of the proposed GOF tests to unknown margins has to be considered in the future.

\section*{Acknowledgments}
The author acknowledge substantial contributions by his colleagues of the research group of Prof. C. Czado at Technische Universit\"at M\"unchen and the support of the TUM Graduate School's International School of Applied Mathematics. A special thanks goes to Peter Song for fruitful discussions. Numerical calculations were performed on a Linux cluster supported by DFG grant INST 95/919-1 FUGG.

\bibliographystyle{model2-names}
\bibliography{references}

\appendix
\section{Rosenblatt's transform for R-vines}
\label{appendix:rosenblatt}

The multivariate probability integral transformation (PIT) of \citet{Rosenblatt1952} transforms the copula data $\boldsymbol{u}=(u_1,\ldots,u_d)$ with a given multivariate copula $C$ into independent data in $[0,1]^d$, where $d$ is the dimension of the data set. 

\begin{definition}[Rosenblatt's transform]
Let $\boldsymbol{u}=(u_1,\ldots,u_d)$ denote copula data of dimension $d$. Further let $C$ be the joint cdf of $\bu$. Then Rosenblatt's transformation of $\bu$, denoted as $\boldsymbol{y}=(y_1,\ldots,y_d)$, is defined as
\[
	y_1 := u_1, \qquad
	y_2 := C(u_2|u_1), \quad
	\ldots \quad
	y_d := C(u_d|u_1,\ldots,u_{d-1}),
\label{def:Rosenblatt}
\]
where $C(u_k|u_1,\ldots,u_{k-1})$ is the conditional copula of $U_k$ given $U_1=u_1,\ldots,U_{k-1}=u_{k-1}, k=2,\ldots,d$.
\end{definition}

The data vector $\boldsymbol{y}=(y_1,\ldots,y_d)$ is now i.i.d.~with $y_i\sim U[0,1]$.
In the context of vine copulas the multivariate PIT is given for the special classes of C- and D-vine in \citet[Algorithm 5 and 6]{Aas_Czado}. It is a straight forward application of the Rosenblatt transformation of Definition \ref{def:Rosenblatt} to the recursive structure of a C- or D-vine copula. Similar, an algorithm for the R-vine can be stated, see Algorithm \ref{alg:PIT}. Here we make use of a similar structured algorithm of \citet{DissmannBrechmannCzadoKurowicka2011} for calculating the log-likelihood of an R-vine copula.
 

In order to perform computations for a general R-vine copula model, it is convenient to use matrix notation  \citep[see][]{MoralesNapolesCookeKurowicka2009,DissmannBrechmannCzadoKurowicka2011}. It stores the edges of an R-vine tree sequence in a lower triangular matrix. For the vine tree sequence of Figure \ref{fig:5dimRvine} this is given by

{\footnotesize\begin{equation}
 M=\begin{pmatrix}
 5 \\
 4 & 4 \\
 3 & 3 & 3 \\
 1 & 2 & 2 & 2 \\
 2 & 1 & 1 & 1 & 1\\
 \end{pmatrix}.
 \label{eq:matrixM}
 \end{equation}}
 
As an illustration for how the R-vine matrix is derived from the tree sequence in Figure \ref{fig:5dimRvine} and vice versa, let us consider the second column of the matrix. Here we have $4$ on the diagonal, and $3$ as a second entry. The set of remaining entries below $3$ is $\{1,2 \}$. This corresponds to the edge $3,4\vert 1,2$ in $T_3$ of Figure \ref{fig:5dimRvine}.
Similarly, the edge $2,4\vert 1$ corresponds to the first and third entry of the second column given the last entry of this row. So the second column of $M$ identifies the edges $3,4\vert 1,2, 2,4\vert 1$ and 1,4. Here we ordered the conditioned set in ascending order.
Further, the diagonal of $M$ is sorted in descending order which can always be achieved by reordering the node labels. The elements of $M$ are denoted by $m_{i,j}, i=1,\ldots,d, j=1,\ldots,i$. From now on, we will assume that all matrices are "normalized" in this way as this allows to simplify notation. 
Similar to the R-vine tree sequence identifying matrix $M$, we can store the corresponding copula families $\mathcal{B}$ and  parameters $\btheta$ in additional lower triangular matrices. 

%


Further, the conditional distributions $C_{j(e)\vert D(e)}$ and $C_{k(e) \vert D(e)}$ are required for the calculation of the log-likelihood function in an R-vine model. Evaluated at a $d$-dimensional vector of observations $(u_1,\ldots, u_d)$, $C_{j(e)\vert D(e)}$ and $C_{k(e) \vert D(e)}$ are the arguments of the copula density $c_{j(e),k(e);D(e)}$ corresponding to edge $e$. We will store these values in the matrices

{\begin{equation}
\label{Vdirect}
V^{direct}=\begin{pmatrix}
\ldots & & & &  \\
\ldots & C(u_4\vert u_{m_{3,2}}, u_{m_{4,2}},u_{m_{5,2}}) & &  &  \\
\ldots &  C(u_4 \vert u_{m_{4,2}}, u_{m_{5,2}}) & C(u_3 \vert u_{m_{4,3}},u_{m_{5,3}}) & &  \\
\ldots & C(u_4\vert u_{m_{5,2}}) & C(u_3\vert u_{m_{5,3}}) &  C(u_2\vert u_{m_{5,4}})   \\
\ldots  & u_4 & u_3 & u_2 & u_1 & \  \\
\end{pmatrix}
\end{equation}
and
\begin{equation}
\label{Vindirect}
V^{indirect}=\begin{pmatrix}
\ldots & & & & & \\
\ldots & C(u_{m_{3,2}}\vert u_{m_{4,2}},u_{m_{5,2}},u_4) & &  & &\ \\
\ldots &  C(u_{m_{4,2}}\vert u_{m_{5,2}}, u_4) & C(u_{m_{4,3}} \vert u_{m_{5,3}},u_3) & & &\ \\
\ldots & C(u_{m_{5,2}}\vert u_4) & C( u_{m_{5,3}}\vert u_3) &  C(u_{m_{5,4}}\vert u_2)  &\ \\
\ldots & u_{m_{5,2}} & u_{m_{5,3}} & u_{m_{5,4}} &  \ \\
\end{pmatrix},
\end{equation}}\noindent
respectively, both of dimension $d\times d$. For computational purposes an additional matrix $\tilde{M}$ is needed to decide whether the arguments of the pair-copula have to be picked from matrix $V^{direct}$ or $V^{indirect}$. For details we refer to 
\citet{DissmannBrechmannCzadoKurowicka2011}.

Algorithm \ref{alg:PIT} now calculates the PIT of an R-vine copula model. 
The vector $\boldsymbol{y} = (y_1,\ldots,y_d)$ stores at the end the transformed PIT variables.\\

\begin{center}
  \captionsetup{style=ruled,type=algorithm,skip=0pt}
  \makeatletter
    \fst@algorithm\@fs@pre
  \makeatother
  \caption{Probability integral transform (PIT) of an R-vine}
  \makeatletter
    \@fs@mid
  \makeatother
  

\begin{algorithmic}[1]
\label{alg:PIT}
	\REQUIRE $d$-dimensional R-vine specification in matrix form, i.e., $M$, $\mathcal{B}$, $\btheta$, where $m_{k,k}=d-k+1, k=1,\ldots,d$, and a set of observations $(u_1,\ldots, u_d)$.
	\STATE Let $V^{\text{direct}} = (v^\text{direct}_{k,i} | i = 1, \ldots, d; k=i,\ldots,d)$.
	\STATE Let $V^{\text{indirect}} = (v^\text{indirect}_{k,i} | i = 1, \ldots, d; k=i,\ldots,d)$.
	\STATE Set $(v^\text{direct}_{d,1}, v^\text{direct}_{d,2}, \ldots, v^\text{direct}_{d,d})
				= (u_d, u_{d-1}, \ldots u_1)$.
	\STATE Let $\tilde{M} = (\tilde{m}_{k,i} | i = 1, \ldots, d; k = i, \ldots, d)$ where $\tilde{m}_{k,i} = \max \{ m_{k,i}, \ldots, m_{d,i} \}$
	for all $i = 1, \ldots, d$ and $k = i, \ldots, d$.
	\STATE Set $y_1 = u_1$
	\FOR[Iteration over the columns of $M$]{$i = d-1, \ldots, 1$}
		\FOR[Iteration over the rows of $M$]{$k = d, \ldots, i+1$}
			\STATE Set $z_1 = v^\text{direct}_{k,i}$ 
			\IF{$\tilde{m}_{k,i} =  m_{k,i}$} 
				\STATE Set $z_2 = v^\text{direct}_{k,(d-\tilde{m}_{k,i}+1)}$.
			\ELSE
				\STATE Set $z_2 = v^\text{indirect}_{k,(d-\tilde{m}_{k,i}+1)}$.
			\ENDIF	
			\STATE Set $v^\text{direct}_{k-1,i} =  h(z_1,z_2 | \mathcal{B}^{k,i}, \theta^{k,i})$ and $v^\text{indirect}_{k-1,i} = h(z_2,z_1 | \mathcal{B}^{k,i},  \theta^{k,i})$.
			\STATE Set $y_{d-k+1} = v^{direct}_{i-1,k}$
		\ENDFOR
	\ENDFOR
	\RETURN $\boldsymbol{y} = (y_1,\ldots,y_d)$
\end{algorithmic}
\makeatletter
    \@fs@post
  \makeatother
\end{center}

\section{Cramér-von Mises, Kolmogorov-Smirnov and Anderson Darling goodness-of-fit test}
\label{appendix:tests}

\subsection{Multivariate and univariate Cram{\'e}r-von Mises and Kolmogorov-Smirnov test}
\label{subsec:CvM}

Already in the third century of 1900 two model specification tests were developed by Cram{\'e}r and von Mises, and by Kolmogorov and Smirnov. Both tests treat the hypothesis that $n$ i.i.d.~samples  $\boldsymbol{y}_1,\ldots,\boldsymbol{y}_n$ of the random vector $\boldsymbol{Y}=(Y_1,\ldots,Y_d)$ follow a specified continuous distribution function $F$, i.e.
\[
	H_0: \boldsymbol{Y}\sim F\quad \text{ versus }\quad H_1: \boldsymbol{Y}\not\sim F.
\]
The general \textbf{multivariate Cram{\'e}r-von Mises (mCvM)} test statistic for a d-dimensional random vector $\boldsymbol{Y}$ is defined as
\begin{equation}
	\text{\textbf{mCvM:}}\quad \omega^2 = \int_{\bbr^d} \left[\hat{F}_n(\boldsymbol{y})-F(\boldsymbol{y})\right]^2dF(\boldsymbol{y}),
	\label{eq:CvM}
\end{equation}
while the \textbf{multivariate Kolmogorov-Smirnov (mKS)} test statistic is 
\begin{equation}
	\text{\textbf{mKS:}}\quad D_n = \sup_{\boldsymbol{y}} |\hat{F}_n(\boldsymbol{y})-F(\boldsymbol{y})|.
	\label{eq:KS}
\end{equation} 
Here $	\hat{F}_n(\boldsymbol{y}) = \frac{1}{n+1}\sum_{j=1}^n \boldsymbol{1}_{\{\boldsymbol{y}_j\leq \boldsymbol{y}\}}$
denotes the empirical distribution function corresponds to the i.i.d.~sample $(\boldsymbol{y}_1,\ldots,\boldsymbol{y}_n)$ of $\boldsymbol{Y}$.\\
The \textbf{univariate cases} for the random variable $Y$ are then denoted by
\begin{equation*}
\begin{split}
	&\text{\textbf{CvM:}}\quad \omega^2 = \int_{-\infty}^{\infty} \left[\hat{F}_n(y)-F(y)\right]^2dF(y)\quad \text {and}\\
	&\text{\textbf{KS:}}\quad D_n = \sup_{y} |\hat{F}_n(y)-F(y)|.
\end{split}
\end{equation*}

\subsection{Univariate Anderson-Darling test}
\label{subsec:AndersonDarling}

\noindent
The \citet{AndersonDarling1954} test, is a statistical test of whether a given probability distribution fits a given set of data samples. It extends the Cram{\'e}r-von Mises test statistics by adding more weight in the tails of the distribution in consideration. Although it has a general multivariate definition we introduce only the univariate case, since only the univariate case is needed in Section \ref{subsec:PIT}. Let $Y$ be a random variable then the null hypothesis of the Anderson-Darling test is again $H_0: Y\sim F(y)$ against the alternative $H_1: Y\not\sim F(y)$.
The general \textbf{univariate Anderson-Darling (AD)} test statistic is defined as
\begin{equation}
	W^2_n = n\int_{-\infty}^{\infty} \left[\hat{F}_n(y)-F(y)\right]^2\psi(F(y))dF(y),
	\label{eq:AD}
\end{equation}
where $\psi(F(y))$ is a non-negative weighting function.
With the weighting function $\psi(u)=\frac{1}{u(1-u)}$ \citet{AndersonDarling1954} put more weight in the tails since this function is large near $u=0$ and $u=1$. Setting the weight function to $\psi(u)=1$ one gets as a special case the Cramér-von Mises test statistic. In the case of uniform margins (\ref{eq:AD}) simplifies to
\begin{equation}
	\text{\textbf{AD:}}\quad W^2_n = n \int_{0}^{1} \frac{\left[\hat{F}_n(y)-y\right]^2}{y(1-y)}dy,\qquad y\in[0,1].
	\label{eq:ADtest}
\end{equation}
%

\section{Model specification for the power study}
\label{appendix:powerstudy}
For the vine copula density (see Equation (\ref{eq:density})) often a short hand notation is used. For this the pair-copula arguments are omitted and denotes only the conditioned and conditioning set. Thus, for the R-vine given in Example \ref{example} we can write
\[
	c_{12345} = c_{1,2}\cdot c_{1,3}\cdot c_{1,4}\cdot c_{4,5}\cdot c_{2,4;1}\cdot c_{1,5;4}\cdot c_{2,3;1,4}\cdot c_{3,5;1,4}\cdot c_{2,5;1,3,4}.
\]
Similarly the considered C- and D-vine copula can be expressed as
\begin{align}
	c_{12345} &= 
	 c_{1,2}\cdot c_{2,3}\cdot c_{2,4}\cdot c_{2,5}\cdot c_{1,3|2}\cdot c_{1,4|2}\cdot c_{1,5|2}\cdot 
 c_{3,4|1,2}\cdot c_{4,5|1,2}\cdot c_{3,5|1,2,4}
 \label{eq:5dimCvine} \\
	c_{12345} &= 
	c_{1,2}\cdot c_{1,5}\cdot c_{4,5}\cdot c_{3,4}\cdot c_{2,5|1}\cdot c_{1,4|5}\cdot c_{3,5|4}\cdot 
 c_{2,4|1,5}\cdot c_{1,3|4,5}\cdot c_{2,3|1,4,5}
 \label{eq:5dimDvine}
\end{align}

\begin{table}[!ht]
	\centering
		\begin{tabular}{llcr}
		\toprule
			Tree & $\mathcal{V}_R^5$ & $\mathcal{B}_R^5(\mathcal{V}_R^5)$ & $\tau$  \\
			\midrule
			1 & $c_{1,2}$ & Gauss & 0.71 \\
				& $c_{1,3}$ & Gauss & 0.33 \\ 
				& $c_{1,4}$ & Clayton & 0.71 \\
				& $c_{4,5}$ & Gumbel & 0.74 \\
			2	& $c_{2,4|1}$ & Gumbel & 0.38 \\
				& $c_{3,4|1}$ & Gumbel & 0.47 \\
				& $c_{1,5|4}$ & Gumbel & 0.33 \\
			3	& $c_{2,3|1,4}$ & Clayton & 0.35 \\
				& $c_{3,5|1,4}$ & Clayton & 0.31 \\
			4 & $c_{2,5|1,3,4}$ & Gauss & 0.13 \\
			\bottomrule
		\end{tabular}
		\caption{Copula families and Kendall's $\tau$ values of the investigated (mixed) R-vine copula model defined by (\ref{eq:5dimRvine}) in the 5-dimensional case.}
		\label{tab:5dimRvine}
\end{table}


\begin{table}[!ht]
	\centering
		\begin{tabular}{llcrllcr}
		\toprule
			Tree & $\mathcal{V}_R^8$ & $\mathcal{B}_R^8(\mathcal{V}_R^8)$ & $\tau$ & Tree & $\mathcal{V}_R^8$ & $\mathcal{B}_R^8(\mathcal{V}_R^8)$ & $\tau$ \\
			\midrule
			1 & $c_{1,2}$ & Joe 	& 0.41 		& 3 & $c_{6,7|1,4}$ & Frank 0.03 	& 7 \\
				& $c_{1,4}$ & Gauss & 0.59		&	  & $c_{1,8|4,7}$ & Gumbel & 0.22 \\
				& $c_{1,5}$ & Gauss & 0.59 		&	  & $c_{3,4|1,6}$ & Gauss & 0.41 \\
				& $c_{1,6}$ & Frank & 0.23 		&   & $c_{2,3|1,6}$ & Gumbel & 0.68 \\
				& $c_{3,6}$ & Frank & 0.19 		&	4 & $c_{6,8|1,4,7}$ & Clayton & 0.17 \\
				& $c_{4,7}$ & Clayton & 0.44 	&	  & $c_{5,7|1,4,6}$ & Gauss & 0.09 \\
				& $c_{7,8}$ & Gumbel & 0.64 	&   & $c_{3,5|1,4,6}$ & Frank & 0.21 \\
			2 & $c_{2,6|1}$ & Clayton & 0.58&   & $c_{2,4|1,3,6}$ & Gumbel & 0.57 \\
				& $c_{1,3|6}$ & Gumbel & 0.44 & 5 & $c_{2,5|1,3,4,6}$ & Joe & 0.25 \\
				& $c_{4,6|1}$ & Frank & 0.11 	&   & $c_{3,7|1,4,5,6}$ & Gumbel & 0.17 \\
				& $c_{4,5|1}$ & Clayton & 0.53&   & $c_{5,8|1,4,6,7}$ & Frank & 0.02 \\
				& $c_{1,7|4}$ & Clayton & 0.29& 6 & $c_{2,7|1,3,4,5,6}$ & Gumbel & 0.31 \\
				& $c_{4,8|7}$ & Gauss & 0.53 	&   & $c_{3,8|1,4,5,6,7}$ & Clayton & 0.20 \\
			3 & $c_{5,6|1,4}$ & Gauss & 0.19& 7 & $c_{2,8|1,3,4,5,6,7}$ & Frank & 0.03 \\
		\bottomrule
		\end{tabular}
		\caption{Copula families and Kendall's $\tau$ values of the investigated (mixed) R-vine copula model in the 8-dimensional case.}
		\label{tab:8dimRvine}
\end{table}

\end{document}